\documentclass[twocolumn,a4paper]{twocolart}
\pdfoutput=1

\usepackage[T1]{fontenc}   % needed for Polish ogonek \k{e} etc. in names
\usepackage{amsmath}
\usepackage{amssymb}
\usepackage{braket}
\usepackage{booktabs}
\usepackage{graphicx}
\graphicspath{{figures/}}
\usepackage{xurl}       % allow long URLs to break across lines
\usepackage{newfloat}   % algorithm float (the algorithm pkg clashes with the class)
\usepackage{algpseudocode}
\DeclareFloatingEnvironment[name=Algorithm,fileext=loa,placement=tbp]{algorithm}
\usepackage{xcolor}
\usepackage{tikz}
\usepackage{quantikz}
\usepackage{subcaption}
\usepackage{amsthm}
\usepackage[numbers,sort&compress]{natbib}
\usepackage{nameref}
\newtheorem{lemma}{Lemma}
\newtheorem{theorem}{Theorem}

\newtheorem{proposition}{Proposition}

\newcommand{\Reach}{\mathcal{R}}        % reachable subspace
\newcommand{\Bgt}{\mathcal{B}}

\begin{document}
\hypersetup{pdfauthor={Karol Bartkiewicz, Patrycja Tulewicz}}

\title{Context-Verified, Error-Budget-Aware Decomposition Selection
for Toffoli Networks}

\author{Karol Bartkiewicz}
\orcid{0000-0002-5355-7756}
\email{karol.bartkiewicz@amu.edu.pl}
\affiliation{Institute of Spintronics and Quantum Information, Faculty of
Physics and Astronomy, Adam Mickiewicz University, 61-614 Pozna\'n, Poland}

\author{Patrycja Tulewicz}
\orcid{0000-0002-7180-4490}
\affiliation{Institute of Spintronics and Quantum Information, Faculty of
Physics and Astronomy, Adam Mickiewicz University, 61-614 Pozna\'n, Poland}
\affiliation{Pozna\'n Supercomputing and Networking Center, IBCH PAS,
Pozna\'n, Poland}

\maketitle

%=======================================================================
\begin{abstract}
%=======================================================================
The cheapest Toffoli (CCX) decompositions (relative-phase and control-reduced
Toffolis) save two-qubit gates, which dominate the error budget of near-term
hardware, but they are correct only in context, since a relative phase is
observable only when superposed branches interfere. We propose quantum chip
performance tuning (QCPT), a compiler that selects the cheapest decomposition
for each Toffoli (stage 1) and admits it only if a per-instance check shows that
it equals the exact gate, up to one phase, on every state that can enter it; a
whole-circuit certificate gates the output. The output acts as the exact circuit
on every valid input, superposed ones included, up to a reported global phase,
which per-state checks do not guarantee. Substituting relative-phase Toffolis
wherever they fit changes the unitary of six of twelve benchmark circuits used
as subroutines, whereas stage 1 changes none and removes $31.0\%$ of the
two-qubit gates. A certified resynthesis of classical circuits of up to ten
qubits raises the saving to $43.1\%$ on the primary suite and to $57.5\%$ on
$20$ small RevLib circuits (in-sample). A bound on the compilation error plus
the per-gate device error separates what the compiler adds from what a
depolarizing two-qubit model of the device adds. In mirror tests on
\texttt{ibm\_aachen}, QCPT (best of four candidates, one of them the Qiskit
circuit) returns more often than Qiskit on all nine circuits; in simulation with
the calibration that drives its selection, on $29$ of $33$ circuits and worse on
none by more than $0.01$. An iterative language-model search with our
certificate obtains no fewer gates on any of nine circuits and needs about $45$
times longer (one run each). QCPT gives no gain without Toffoli structure.
Certified context-dependent selection can thus let near-term devices run larger
Toffoli circuits at a given success probability, with a gain that varies by
family.
\end{abstract}

%=======================================================================
\section{Introduction}
%=======================================================================
On near-term noisy intermediate-scale quantum (NISQ) hardware, circuit failure is dominated by two-qubit
gates, whose infidelities exceed single-qubit ones by roughly an order of
magnitude (a factor $7$--$47$ on current platforms, computed from the published two- and single-qubit error rates and the vendors' specifications~\cite{QuantinuumH2,IBMHeronR2,GoogleWillow,IonQForte,IonQAria}),
so the fidelity of a compiled circuit is governed by the number and type of
its two-qubit gates~\cite{Preskill2018}. The Toffoli (CCX) gate is the basic
non-Clifford primitive of quantum arithmetic, oracles, and reversible
logic, and it admits many decompositions trading ancillae, depth, and
two-qubit-gate count. It also appears in quantum state-resetting protocols~\cite{QuantumResetting2025}. However, the cheapest known decompositions are not exact CCX implementations.
Relative-phase Toffolis~\cite{Maslov2016} realize the Toffoli permutation
up to a diagonal phase on a subset of basis states, and approximate
Toffolis realize it only to within a tolerance. Both are strictly cheaper in
two-qubit gates than any exact decomposition but correct only in context,
i.e., the residual phase has to be uncomputed downstream or the error has to
stay within tolerance on the states that reach the gate. Can a compiler use such gates and still preserve the function of the circuit, superposed inputs included? We call a compiler that selects and certifies cheaper but correct gates on fixed hardware
\emph{quantum chip performance tuning} (QCPT; not to be confused with ``quantum circuit transformation'', the qubit-mapping task~\cite{SABRE2019,Zhou2020MCTS,Li2021QCT}), after the performance tuning of a car, which gets more
out of the same engine; unlike an engine remap, every change here carries a
certificate of correctness (figure~\ref{fig:concept} summarizes how QCPT works, figure~\ref{fig:dyno} shows the analogy for the measured gain).

\begin{figure*}[t]
\centering
\includegraphics[width=0.92\textwidth]{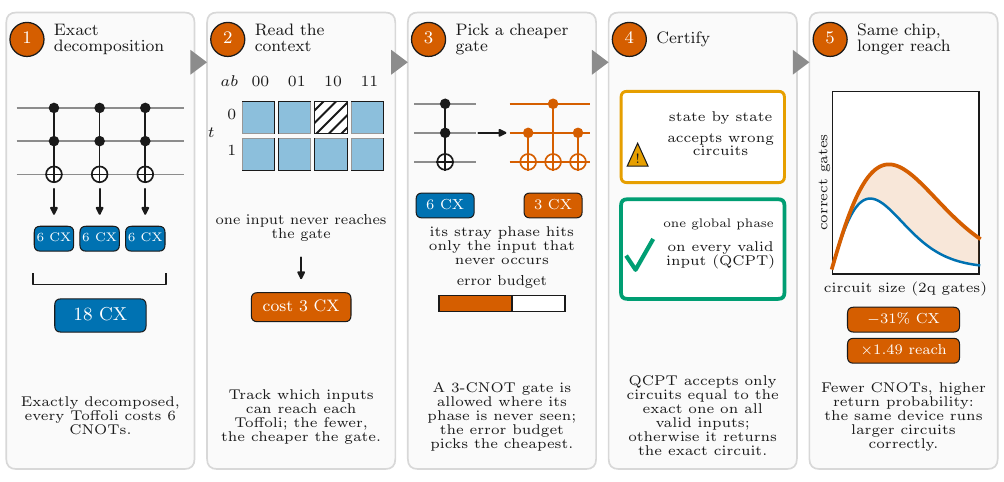}
\caption{Quantum chip performance tuning (QCPT) in five steps. (1)~Decomposed exactly, every Toffoli costs six CNOTs. (2)~QCPT tracks which inputs can reach each Toffoli; the fewer, the cheaper the gate it may use. (3)~A three-CNOT gate whose stray phase falls on an input that never occurs replaces the exact gate as far as the error budget allows. (4)~A certificate accepts the result only if it equals the exact circuit, up to one global phase, on every valid input; a rejected result falls back to the exact circuit. (5)~The same chip then runs larger circuits at a given success probability: the circuit \texttt{4mod5-bdd\_287} returns on \texttt{ibm\_aachen} with probability $0.46$ (Qiskit) and $0.62$ (QCPT), or $0.60$ and $0.78$ with the ancillae post-selected on $\ket{0}$. The curves show the number of correct gates for four measured modular incrementers (three parameters per curve; Section~\ref{sec:eval-hw} gives the sensitivity); the load scale $x_0$ moves from $241$ to $361$ gates ($\times1.49$).}
\label{fig:concept}
\end{figure*}

Existing tools (Table~\ref{tab:axes}) either optimize the gate count on annotated or inferred states without a hardware fidelity model~\cite{RPO2021,Haner2020,QCP2023}, select exact Toffoli variants by topology and neighboring gates~\cite{QContext2023}, prove only context-free equivalences~\cite{VOQC2021,Quartz2022}, or replace relative-phase Toffolis where the phases cancel by construction~\cite{Maslov2016,KurodaYamashita2022}. None of these optimizers verifies, per instance, a substitution that is correct only in context, or handles phases that cancel only jointly. A pattern matcher that substitutes a relative-phase Toffoli wherever one fits (count-greedy substitution) corrupts circuits silently.

In this paper, we treat Toffoli optimization with context-dependent decompositions as a soundness problem. Three ideas carry the paper, and each has one main result.
\begin{enumerate}
\item \textbf{Certified context.} Admissibility is one operator identity with a single phase on the reachable subspace; it covers decompositions that are exact on that subspace, compute/uncompute pairs, and windows whose phases cancel only jointly, such as the blocks of the Cuccaro adder. It is sound for superposed inputs (\S\ref{sec:soundness}), which a check state by state is not. The conditions predict what the certificate admits, and the certificate decides the output. Substituting relative-phase Toffolis wherever they fit changes the unitary of $6$ of $12$ benchmark circuits used as subroutines (as stand-alone programs they remain correct); QCPT changes none. Unlike the matched relative-phase pairs inside the implementation of Qiskit's V-chain multi-controlled X and of Qrisp's automatic uncomputation~\cite{Qrisp2023}, QCPT admits such decompositions in an arbitrary flat gate list and decides each admission per instance.
\item \textbf{Certified search.} One whole-circuit certificate, with a fallback to the last certified circuit, gates both the per-gate selection of the cheapest decomposition (\S\ref{sec:method}) and the resynthesis of small classical circuits (\S\ref{sec:method-resynth}). Stage 1 removes $31.0\%$ of the two-qubit gates of the primary suite, QCPT with both stages $43.1\%$, and $57.5\%$ on $20$ small RevLib circuits.
\item \textbf{Error accounting and hardware.} A bound $\epsilon+\sum_e\eta_e$ separates what compilation adds from what a depolarizing two-qubit model of the device adds (\S\ref{sec:method-account}); for approximate rotations we use a calibrated variant. In mirror tests on \texttt{ibm\_aachen} QCPT, taking the best of four candidate compilations, returns more often than the Qiskit transpilation of the same circuit on all $9$ circuits and on all $16$ circuits of four families (on $15$ of them by more than three standard errors), and it extends the load at which the chip still runs a family correctly, by a factor of the order of $1$--$2$ (three-parameter fits to six to eight points, \S\ref{sec:eval-hw}). A Fourier adder with dropped small rotations beats the exact one on average (\S\ref{sec:eval-approx}).
\end{enumerate}
We review prior work in section~\ref{sec:related}, introduce the taxonomy and the error model in section~\ref{sec:background}, present and prove QCPT in section~\ref{sec:method}, evaluate it in section~\ref{sec:eval} and conclude in section~\ref{sec:conclusion}; the appendices give the hardware error model, the device characterization, the benchmark suite and the cost bounds.

%=======================================================================
\section{Related work}
\label{sec:related}
%=======================================================================
To our knowledge, no prior approach combines the three components of
QCPT, i.e., context-aware selection, an error-budget objective and
per-instance verification (Table~\ref{tab:axes}).

\begin{table*}[t]
\centering
\caption{Comparison of this work with prior approaches along the three axes.
$^{\dagger}$template/pattern-based, with no per-instance check.
$^{\ddagger}$context-free (unconditional) equivalence only.
$^{\S}$whole-unitary fidelity distance, not a per-gate two-qubit-infidelity
budget. RPO: relaxed peephole optimization; QCP: quantum constant propagation; AQC: approximate quantum compiling; 2q: two-qubit.}
\label{tab:axes}
\small
\setlength{\tabcolsep}{4pt}
\begin{tabular}{lcccc}
\toprule
Approach & Context-aware & Error-budget obj. & Per-instance verified & Optimizes \\
\midrule
QContext~\cite{QContext2023}                              & $\checkmark$            & $\times$              & n/a (exact variants)  & gate count \\
RPO; QCP~\cite{RPO2021,QCP2023}                           & $\checkmark$ (known states) & $\times$          & exact on known states & gate count \\
Maslov; Kuroda--Yamashita~\cite{Maslov2016,KurodaYamashita2022} & $\checkmark^{\dagger}$ & $\times$ & $\times$ & $T$-/gate count \\
VOQC~\cite{VOQC2021}; Quartz~\cite{Quartz2022}           & $\times$                & $\times$              & $\checkmark^{\ddagger}$ & gate count \\
AQC~\cite{AQC2021}                                       & $\times$                & $\checkmark^{\S}$     & $\times$              & fidelity \\
NuOp; Bowman~\cite{NuOp2021,Bowman2023}                  & $\times$                & $\checkmark$         & $\times$              & fidelity \\
\midrule
\textbf{This work}                                       & $\checkmark$            & $\checkmark$          & $\checkmark$          & 2q infidelity \\
\bottomrule
\end{tabular}
\end{table*}

\paragraph{Context-aware selection.} QContext~\cite{QContext2023} selects from $32$ exact Toffoli variants the one whose boundary CNOTs cancel best, and Nam et al.~\cite{Nam2018} choose between the polarities of the exact decomposition to minimize the $T$-count. Relative-phase Toffolis and pairs of them go back to Barenco et al.~\cite{Barenco1995}; Maslov~\cite{Maslov2016} and Kuroda \& Yamashita~\cite{KurodaYamashita2022} replace Toffolis by relative-phase Toffolis and repair the residual phases by construction, as do the technology mappings of reversible circuits~\cite{Abdessaied2016,Selinger2013} ; ReVerC~\cite{ReVerC2017} compiles Toffoli networks with verified uncomputation but without relative phases. PennyLane~\cite{PennyLane2018} rewrites a fixed four-qubit relative-phase Toffoli pattern~\cite{PennyLaneRPT}, and Qiskit's \texttt{n\_dirty\_i15} synthesis optionally applies the construction of Iten et al.~\cite{Iten2016} up to a relative phase; neither checks admissibility in the surrounding circuit. Onodera et al.~\cite{Multilevel2025} lower the two Toffolis of a nested pair with conjugate decompositions that cancel at the next level of a compilation flow for Hamiltonian-simulation circuits, without a reachability condition or a certificate, and Amy and Ross~\cite{AmyRoss2021} generalize relative-phase and measurement-assisted constructions. Synthesis of reversible and Clifford+$T$ circuits minimizes gate count, depth or $T$-count through transformation-based~\cite{MillerMaslovDueck2003}, cost-optimal~\cite{Shende2003}, meet-in-the-middle~\cite{AmyMMR2013} or $T$-count-optimal~\cite{HeyfronCampbell2019} algorithms, but it demands exact equivalence and uses no hardware error model. None targets a two-qubit infidelity budget or checks admissibility in an arbitrary surrounding circuit.

\paragraph{State-aware optimization.} Our reachable-subspace criterion is closest to optimization on statically known states: relaxed peephole optimization~\cite{RPO2021}, assertion-based optimization~\cite{Haner2020}, quantum constant propagation~\cite{QCP2023}, AQCEL~\cite{AQCEL2022}, and relational and abstract analyses~\cite{YuPalsberg2021,AmyLunderville2025,Abstraqt2023} remove controls and gates that are exact on the inferred states. Our control drop is an instance of this idea, but these works do not admit decompositions inexact on the reachable states, such as relative-phase Toffolis whose phase must stay unobservable downstream, and do not rank alternatives by two-qubit infidelity.

\paragraph{Verification.} VOQC~\cite{VOQC2021} and Quartz~\cite{Quartz2022} prove context-free, unconditional equivalences, CertiQ~\cite{CertiQ2019} and Giallar~\cite{Giallar2022} verify the passes of Qiskit, and search-based optimizers use context-free rewrites, QUESO~\cite{Queso2023} with probabilistically verified rules and GUOQ~\cite{GUOQ2025} with rewriting plus approximate resynthesis. The equivalence that we certify depends on the surrounding circuit. The MQT quantum circuit equivalence checker (QCEC)~\cite{QCEC2021}, which our final certificate calls, extends to dynamic circuits with non-unitary operations~\cite{BurgholzerNonunitary2022} and is complemented by checking based on the ZX-calculus~\cite{Peham2022}; the flow-verification use of QCEC goes back to Ref.~\cite{BurgholzerWille2020Qiskit}, and reversible-miter checking to Ref.~\cite{YamashitaMarkov2010}. Observational equivalence (\S\ref{sec:soundness}) is close to the partial equivalence of Chen et al.~\cite{PartialEq2022}; subroutine equivalence also fixes the relative phases of the output, and we check it for each substitution inside a selector. Path sums~\cite{Amy2019}, equivalence checking with measurements~\cite{Hong2022} and automata-based verification~\cite{AutoQ2023} are closest in form to our condition on $\Reach(g)$.

\paragraph{Fidelity objectives and uncomputation.} Approximate quantum compiling~\cite{AQC2021} minimizes a fidelity distance for the whole unitary and without context; qubit mapping~\cite{SABRE2019} and noise-aware mapping~\cite{Murali2019,Tannu2019,Nishio2020,tket2020}, NuOp~\cite{NuOp2021}, native-gate Toffolis~\cite{Bowman2023} and optimal Clifford+$T$ synthesis of rotations~\cite{RossSelinger2016} optimize single gates or the mapping, and ZX and phase-polynomial methods~\cite{Duncan2020,KissingerWetering2020,PyZX2019,AmyMM2014,Nam2018} reduce $T$-count by unconditional rewriting. Unqomp~\cite{Unqomp2021}, Reqomp~\cite{Reqomp2024}, Silq~\cite{Silq2020} and modular uncomputation~\cite{Venev2024} synthesize or type-check uncomputation, and Qrisp~\cite{Qrisp2023} uses phase-tolerant gates inside uncomputed variables, so condition (C) holds there by construction; QCPT detects and certifies such pairs in a flat gate list. Measurement-based uncomputation~\cite{Gidney2018,Jones2013,Paler2022,Kole2026} is cheaper than any unitary pair and the existence of an uncomputation is coNP-hard~\cite{LiuZhouMeng2026}; we leave both it and conditionally clean ancillae~\cite{CleanAncillae2025} (\S\ref{sec:conclusion}) out of the unitary taxonomy.

%=======================================================================
\section{Background}
\label{sec:background}
%=======================================================================

\subsection{Decomposition taxonomy}
\label{sec:bg-taxonomy}
The exact CCX without ancillae uses $6$ two-qubit gates and the relative-phase (Margolus) Toffoli~\cite{Margolus1990,Barenco1995,DiVincenzo1998,Maslov2016} only $3$, at the expense of a diagonal residual phase; both counts are optimal for CNOT gates~\cite{ShendeMarkov2009,Song2004}, whereas five two-qubit gates of arbitrary type suffice for the exact Toffoli~\cite{Yu2013}. A relative-phase Toffoli $\tilde U=D\,U_{\mathrm{CCX}}$ acts as CCX in a context if $D$ is canceled downstream, e.g., when the target is later uncomputed, or if $D$ reduces to one global phase on the images $U_{\mathrm{CCX}}\ket\psi$ of the states $\ket\psi$ that reach the gate. An approximate Toffoli with $\lVert\tilde U-U_{\mathrm{CCX}}\rVert\le\epsilon$ is valid when the error stays within tolerance on those states; we do not use this option. Each gate $g$ has a candidate set $\mathcal D(g)$ drawn from the list below.

The candidates are the exact CCX ($6$ CX); the Margolus gate (phase $-1$ on $abt=100$, $3$ CX); the mirrored Margolus gate (phase $-1$ on $101$, $3$ CX; exact when $t$ enters in $\ket0$ or holds $a\wedge b$); the control drop ($0$--$1$ CX); and the approximate variants ($0$--$1$ CX, error $\le\epsilon$). The last two act on the reachable subspace $\Reach(g)$ (\S\ref{sec:method-context}).

\subsection{Error model}
\label{sec:bg-fidelity}
Each two-qubit gate $e$ has an infidelity $\varepsilon_e=1-F_{\mathrm{avg}}(e)$ from the calibration. The infidelity of a decomposition $d$ is $\mathcal I(d)=\sum_{e\in\mathrm{2q}(d)}\varepsilon_e$, the first-order form of $1-\prod_e(1-\varepsilon_e)$ for independent depolarizing errors (up to the constant factor $5/4$ between process and average infidelity, which does not change the ranking), and the \emph{error budget} of a circuit is $\Bgt=\sum_g\mathcal I(d_g)$. Single-qubit errors are a lower-order term, which does not change the ranking as long as a substitution adds fewer than $p_{2q}/p_{1q}$ single-qubit gates per two-qubit gate that it saves (\S\ref{sec:method-select}). If all two-qubit gates share one rate, as in all experiments below, both measures are monotone in the two-qubit count; for heterogeneous $\varepsilon_e$ our implementation minimizes the multiplicative form. We switch off idle errors, static $ZZ$ cross-talk and an excess two-qubit error (Appendix~\ref{app:hw}) unless stated otherwise.

\subsection{Exact equivalence checking}
\label{sec:bg-eqcheck}
Per-instance verification reduces to equivalence checks of the whole circuit against its exact decomposition, or of a single gate against CCX on a reachable subspace. We check whole circuits by the dense unitary for small widths and with decision diagrams~\cite{QCEC2021} for larger ones; no result relies on ZX checking. Section~\ref{sec:method-verify} defines the obligation and \S\ref{sec:eval-scal} its scalability.

%=======================================================================
\section{Method}
\label{sec:method}
%=======================================================================
QCPT takes a Toffoli network and a hardware error model and returns a decomposed circuit with an estimated error budget. It works in two stages (figure~\ref{fig:flow}); we write ``stage 1'' for the per-gate selection and ``QCPT'' for the compiler with both stages, and say so where a number refers to stage 1 alone. Stage 1 selects a decomposition for each Toffoli in four steps: structural reduction by rewrite rules (Appendix~\ref{app:rules}: annihilation of adjacent CCX pairs, fan-out common-subexpression elimination and permutation-equivalence simplification), context analysis (\S\ref{sec:method-context}), error-budget selection (\S\ref{sec:method-select}) and per-instance verification (\S\ref{sec:method-verify}). Stage 2 is a certified resynthesis of small classical circuits (\S\ref{sec:method-resynth}). Algorithm~\ref{alg:pass} outlines the whole procedure as a logical composition: the code admits a candidate to $\mathcal A(g)$ only after its check passes, and the release implements stage 2 in driver scripts. We switch structural reduction off in all results of this paper; the soundness result (\S\ref{sec:soundness}) covers the three steps after it.

\begin{algorithm}[t]
\caption{Context-verified, budget-aware selection}
\label{alg:pass}
\begin{algorithmic}[1]
\Require Toffoli network $C$; error model $\{\varepsilon_e\}$;
         tolerance $\epsilon$; budget bound $\Bgt_{\max}$
\State $C \gets \textsc{StructuralReduce}(C)$
       \Comment{pair annihilation, common-subexpression elimination, permutation equivalence}
\For{each Toffoli $g \in C$}
  \State $\mathcal{A}(g) \gets \textsc{AdmissibleDecomps}(g, C, \epsilon)$
         \Comment{\S\ref{sec:method-context}}
\EndFor
\State $\{d_g\} \gets \textsc{BudgetSelect}(\{\mathcal{A}(g)\},
        \{\varepsilon_e\}, \Bgt_{\max})$
        \Comment{\S\ref{sec:method-select}}
\For{each Toffoli $g$ with chosen $d_g$}
  \If{$d_g$ context-dependent \textbf{and not} \textsc{Verify}$(g, d_g, C)$}
     \State fall back to cheapest certified $d_g$
            \Comment{\S\ref{sec:method-verify}}
  \EndIf
\EndFor
\If{\textbf{not} \textsc{CertifyWhole}$(C)$} \Comment{\S\ref{sec:method-verify}}
  \State \Return all-exact $C$ \Comment{fail closed}
\EndIf
\If{$C$ is classical on $\le 10$ qubits}
  \State $C' \gets \textsc{Resynthesize}(C)$ \Comment{\S\ref{sec:method-resynth}}
  \If{\textsc{CertifyWhole}$(C')$ \textbf{and} cost$(C')<$ cost$(C)$} $C \gets C'$ \EndIf
\EndIf
\State \Return decomposed $C$, certified budget $\Bgt = \sum_g \mathcal{I}(d_g)$
\end{algorithmic}
\end{algorithm}

\begin{figure*}[t]
\centering
\includegraphics[width=0.92\textwidth]{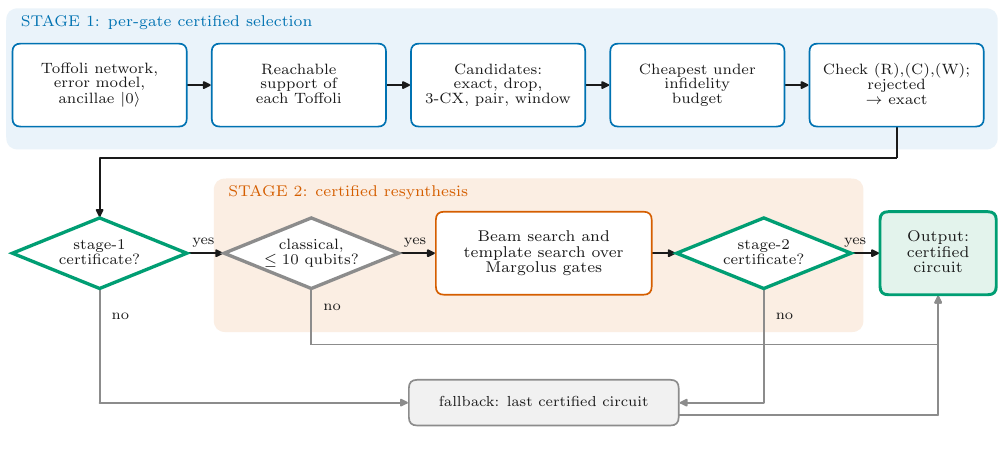}
\caption{\textbf{Flow of QCPT.} Stage 1 selects a decomposition for every Toffoli and certifies the circuit. For classical circuits on at most ten qubits, stage 2 searches for a cheaper circuit with the same net function on the valid inputs. Every output is certified; if a certificate fails, QCPT returns the last certified circuit.}
\label{fig:flow}
\end{figure*}

\subsection{Context analysis}
\label{sec:method-context}
For each Toffoli $g$ we compute the set $\mathcal{A}(g)\subseteq
\mathcal{D}(g)$ of admissible decompositions (\S\ref{sec:soundness}).
Exact decompositions are always admissible (condition (E)). A relative-phase
decomposition $\tilde U=D\,U_{\mathrm{CCX}}$, with $D$ diagonal, is
admissible if its phase cancels against a matching relative-phase Toffoli later in the
circuit, as in compute/uncompute pairs~\cite{Barenco1995} (condition (C)), if
it belongs to a mirrored pair whose segment is certified to act exactly
(condition (W)), as in the majority (MAJ) and unmajority-and-add (UMA) blocks of the Cuccaro adder~\cite{Cuccaro2004}, or if its phase is constant on the states that
reach $g$ (condition (R)).

For controls $a,b$ and target $t$ the three-CNOT Margolus gate applies its
phase $-1$ to the branch $abt=100$, and the mirrored Margolus gate (the four $R_y$
angles negated) applies it to $abt=101$. With $R_y(\theta)=e^{-i\theta Y/2}$, the circuit
\begin{multline*}
R_y(\tfrac{s\pi}{4})\,\mathrm{CX}(b,t)\,R_y(\tfrac{s\pi}{4})\,\mathrm{CX}(a,t)\\
R_y(-\tfrac{s\pi}{4})\,\mathrm{CX}(b,t)\,R_y(-\tfrac{s\pi}{4})
\end{multline*}
in time order on the target $t$ is $D\,U_{\mathrm{CCX}}$ with $D=-1$ on $abt=100$ for $s=-1$ (Margolus gate) and on $101$ for $s=+1$ (mirrored gate). Hence, (R) admits the mirrored
relative-phase Toffoli without any matching gate whenever $t$ enters in $\ket0$, as for an
AND written into a clean ancilla, or holding $a\wedge b$, as in its uncompute.
A fourth condition (U), that no later gate lets the phase interfere before
the final measurement, is valid only for a complete program, so QCPT
applies it only in an opt-in \emph{program mode} (with (U)), whereas the default \emph{subroutine mode} uses only (E), (R), (C) and (W); program mode is not used for any main result
(figure~\ref{fig:cost}a, \S\ref{sec:eval-scal}).

A specialized or approximate decomposition, e.g., a Toffoli with one control
dropped, is admissible if it agrees with $U_{\mathrm{CCX}}$ on $\Reach(g)$ up
to one global phase (condition (R)), or within tolerance $\epsilon$ in
operator norm (the bounded-approximate variant of (R)). Both conditions are
weaker than the global condition of approximate
compilation~\cite{AQC2021}, so they admit cheaper decompositions. Both
are imposed on the whole subspace $\Reach(g)$ with a single phase
(\S\ref{sec:soundness}).

\paragraph{Constructing $\Reach(g)$.}
We compute a reachable support $S(g)$, a set of computational-basis strings,
by sound forward propagation from the valid inputs $S_0$ (free data qubits
range over $\{0,1\}$, ancillae are pinned to $0$). For each gate $h$
preceding $g$ in topological order, $S_{i+1}$ is the set of basis strings to
which $h$ maps an element of $S_i$ with nonzero amplitude. Classical
reversible gates (X, CX, CCX, SWAP) permute basis states, so their image is
exact and cheap, and diagonal gates leave the support unchanged (the implementation treats every gate other than X, CX, CCX and SWAP as freeing its qubits, which is sound but looser). A
relative-phase Toffoli $D\,U_{\mathrm{CCX}}$, and a control-drop admitted
under (R), which acts as $U_{\mathrm{CCX}}$ on every state that reaches it,
map the support as $U_{\mathrm{CCX}}$ does. After a bounded-approximate
substitution QCPT treats the target of the substituted gate as free, and every qubit that later gates couple to it, because the output can leave the box there. Once a gate
that creates superposition (e.g., a Hadamard) occurs in the prefix, $S$
becomes a coordinate box in which qubits fixed by an invariant (an ancilla
provably at $0$, a constant propagated through classical gates) stay fixed
and all others range freely. Every step over-approximates, so checking (R)
on $S(g)$ is sound (\S\ref{sec:soundness}). We set
$\Reach(g)=\mathrm{span}\{\ket x : x\in S(g)\}$. If the prefix of $g$ is
classical reversible and at most $22$ input bits are free, we propagate all
valid inputs exactly as an array of integers, which lets (R) admit Toffolis
writing into clean ancillae in the middle of an adder.

\subsection{Error-budget selection}
\label{sec:method-select}
Given the admissible sets $\{\mathcal{A}(g)\}$ and the per-gate
infidelities $\mathcal{I}(d)$ from \S\ref{sec:bg-fidelity}, we choose
$d_g \in \mathcal{A}(g)$ so as to minimize the total budget
$\Bgt = \sum_g \mathcal{I}(d_g)$, which can be optionally subject to
ancilla-width or depth constraints.

\paragraph{Greedy selection.}
With the context windows fixed, $\mathcal{A}(g)$ and $\mathcal{I}(d)$ are determined, $\mathcal{A}(g)$ does not depend on $d_{g'}$ for $g'\neq g$, and $\Bgt$ is a sum of per-gate terms, so the greedy rule $d_g^\star=\arg\min_{d\in\mathcal{A}(g)}\mathcal{I}(d)$ minimizes it and needs $O\!\big(\sum_g|\mathcal{A}(g)|\big)$ time beyond the per-gate verification. Separability fails in three cases: the two substitutions of a compute/uncompute pair are admissible only jointly; admission under (W) or (U) depends on the order, because the implementation verifies each candidate against the cumulative circuit (by Lemma~\ref{lem:support} this does not occur for (E), (R) and (C)); and a physical ancilla reused across gates lets a choice at $g$ change what is admissible (or feasible within the width bound) at a later $g'$. The greedy rule is therefore optimal when pairs and windows are treated as single units, admission does not depend on the order, the supports are fixed and no width or depth constraint is active; in general it is a heuristic. The problem then becomes an integer linear program (ILP) with a binary variable per $(g,d)$ pair, a one-hot constraint per gate, linear width and depth constraints and objective $\Bgt$; the present benchmarks allocate ancillae per compute/uncompute scope and do not need it, and we have not compared the greedy result with an exhaustive optimum.

\paragraph{Selection in practice.}
A relative-phase Toffoli ($3$ two-qubit gates) or an exact-on-reachable control drop ($1$ or $0$) never increases the two-qubit count. The budget counts two-qubit gates only, so the choice does not depend on $p_{2q}/p_{1q}$; including single-qubit errors gives the same choice if the substitution saves at least one two-qubit gate and adds fewer than $p_{2q}/p_{1q}$ single-qubit gates per gate saved (or none), which is at least $7$ over the five platforms, so the selected mix does not change for $p_{2q}/p_{1q}$ from $7$ to $47$ (\S\ref{sec:eval-struct}). The budget therefore coincides with the minimum two-qubit count, and we do not claim ratio-adaptive selection. It affects the choice in two cases: it orients a Margolus gate so that its doubled control-target pair lies on the cheaper edge of a calibrated device (the reoriented gate is checked again, Appendix~\ref{app:orient}), and it admits a bounded-approximate decomposition (\S\ref{sec:soundness}) only if its reduction of the two-qubit infidelity exceeds the approximation budget $\epsilon$ that it uses. This rule is a heuristic (it compares an infidelity with a norm tolerance) that no experiment triggers for Toffoli decompositions. Two other approximation criteria are separate heuristics: the conservative bound $\epsilon^2+\sum_e\eta_e$ of \S\ref{sec:method-account}, and the ranking by the calibrated cost for rotations in \S\ref{sec:eval-approx}, which we apply to the Fourier adder (Appendix~\ref{app:time} gives a further criterion for running time).

\subsection{Per-instance verification}
\label{sec:method-verify}
In the verification stage QCPT checks each candidate against the condition of
\S\ref{sec:soundness} that admits it, regardless of the rule that proposed
it, and the output circuit is then certified as a whole. For pairs (C), a
structural window check establishes that no gate between the two Toffolis
writes a control or the target, so the diagonal phase commutes with the
window, and the substitution is confirmed by the unitary equivalence, up to a
global phase, of the whole circuit and the all-exact circuit. The QCPT library compares the two circuits block by block: it aligns them gate by gate and checks every block of difference, up to a phase, on the qubits that it touches, after a common stretch whose gates are diagonal on the qubits that they share with the block has been slid past it. This is rigorous whatever the alignment is, because a wrong alignment can only make a block fail, and its cost does not depend on the width of the circuit. Blocks that touch more than $12$ qubits, and equalities that hold only on the subspace of the pinned ancillae, are undecided and fall back to the dense check, which rejects a pair in a circuit wider than $12$ qubits; the evaluation driver of the scale suite confirms such pairs with QCEC~\cite{QCEC2021} above $10$ qubits ($12$ with pinned ancillae). An opt-in Monte Carlo certificate propagates $K=24$ random vectors through both circuits and accepts a pair whose Frobenius deviation exceeds $10^{-8}$ with probability below $(0.1\,e^{0.9})^{K-1}<10^{-14}$; a counterexample is deterministic.

Windows (W) need a separate check, as the structural rule rejects a pair as
soon as a gate in between writes a control or the target, even if the writes
are undone symmetrically. For each mirrored pair, innermost first, we compare
the dense unitaries of the segment on the qubits that it touches, with and
without the substitutions, up to a global phase, so the cost depends only on the
window width. For (R), we project $S(g)$ onto the three qubits of $g$ and
check $\tilde U_dP=e^{i\theta}U_{\mathrm{CCX}}P$ on the span of the projected
strings, with one phase $\theta$, as an $8\times 8$ matrix identity (in
operator norm for the approximate variant), trying the Margolus gate and
the mirrored Margolus gate for a standalone Toffoli and the reduced gate for a
control-drop. The projected span contains $\Reach(g)$, so the check is sound
by \S\ref{sec:soundness}.

The whole-circuit certificate checks
$U_{C'}P_{\mathrm{in}}=e^{i\theta}U_CP_{\mathrm{in}}$ on the input subspace,
which covers every superposition of valid inputs. QCPT builds the columns of both unitaries on the input subspace for up to $12$ qubits and above that calls QCEC~\cite{QCEC2021}
with the pinned qubits declared as ancillae initialized to $\ket0$, which
decides the same identity. We run it with the ZX checker and the elision of permutations switched off and with the trace threshold $10^{-12}$ (with its defaults the checker accepts circuits that differ on an ancilla or by a small rotation), so a verdict holds up to this numerical tolerance and the time limit. This is stronger than the partial equivalence of
Chen et al.~\cite{PartialEq2022}, which compares only the output measurement
statistics. We accept only the verdicts ``equivalent'' and ``equivalent up to
a global phase'' as a certificate. A candidate that fails its check is
rejected, and if the whole-circuit certificate fails or cannot be established
within its time limit, QCPT returns the all-exact circuit. Hence the
output is sound by construction, up to the numerical tolerance of the checker and its time limit (\S\ref{sec:soundness}).

\subsection{Certified resynthesis}
\label{sec:method-resynth}
After stage 1 the circuit still contains structure that no per-gate rule removes, e.g., a CX that undoes the work of a neighboring Margolus gate. For a classical circuit, i.e., one built from \texttt{x}, \texttt{cx} and \texttt{ccx}, only its net function on the valid inputs matters, so we resynthesize it. Let $F$ map each valid input $x\in\mathcal V$ (ancillae pinned to $\ket0$) to the basis string that the circuit produces. We search for a circuit $C'$ with $C'\ket x=e^{i\varphi}\ket{F(x)}$ for all $x\in\mathcal V$ and a single global phase $\varphi$ (the search uses phases in $\mathbb Z_8$ and the certificate then fixes $\varphi$, which we report for controlled calls), i.e., the equivalence of \S\ref{sec:method-verify}.

\paragraph{Library.} The library contains \texttt{x}, \texttt{cx}, \texttt{cz} and the phase gates \texttt{t}, \texttt{s}, \texttt{z} with their inverses (phases in $\mathbb Z_8$, no two-qubit cost). It also contains the exact \texttt{ccx} and \texttt{ccz} ($6$ CX) and the four relative-phase Toffolis of Margolus type, each with a single phase $-1$ on one basis state, Qiskit's \texttt{rccx} in both orders of the controls (its residual phases are $-i$, $+i$ and $-1$ on $110$, $111$ and $101$), and their inverses ($3$ CX each). The \emph{parity Margolus gate}, whose middle control is the parity of two qubits, is realized with one CNOT per qubit at $4$ CX instead of the $5$ of the conjugated gate (Appendix~\ref{app:rules}). Finally, the library has the sign patterns of the four $R_y(\pm\pi/4)$ rotations, two of which no other library gate reaches at this cost. The action of every gate on the basis states is a table, so a state is a byte string indexed by the valid inputs.

\paragraph{Template search.} Instead of a beam guided by a distance to the target, whose pruning discards the intermediate states that lead to the shortest circuits, we enumerate circuits of the form $F_1\,G_a\,F_m\,G_b\,F_2$ exhaustively by meeting in the middle (up to hash collisions, which the certificate catches), where $F_\bullet$ are short sequences of free gates (\texttt{x}, \texttt{cx}) and $G_a,G_b$ are Margolus gates (or their mirrors). The forward states $F_1G_a$ and the backward states $F_2^{-1}G_b^{-1}$, started from the target, are hashed modulo the global phase and every middle $F_m$ is looked up in the sorted backward hashes. The leading \texttt{x} and \texttt{cx} gates of the exact circuit are absorbed into the inputs. A beam search over the whole library, with the cost bound tightened after every circuit found, covers circuits outside this shape.

\paragraph{Implementation.} The search core is written in C++17 and uses nothing beyond the standard library, POSIX threads and AVX2 intrinsics. A gate acts on at most four qubits, so its action on a state is a table of $16$ entries, and the beam search scores the $32$ valid inputs of a candidate at once with byte shuffles (AVX2; $16$ for circuits of $9$ and $10$ qubits); the template search and the certificate use the same tables. The members of a race (template searches and beams with different weights and widths) share one thread pool and run in a single process; when one member finds a cheaper circuit, the others are canceled and the final circuit is certified and the next round starts from that circuit. Python only prepares the problem and calls the core. The template search runs for at most seven qubits; the beam search and the windows were run up to twelve.

\paragraph{Windows with reachable context.} For circuits beyond ten qubits, or whose structure the template search does not cover, we resynthesize contiguous windows of at most five qubits. The replacement of a window has to agree with it only on the local states that the valid inputs of the whole circuit can bring to its input, namely the projection onto the window qubits of the states reached by the prefix, with ancillae starting in $\ket{0}$, and with \emph{one} phase common to all these local states (a phase chosen per local state is not sufficient, as for condition (R)); this is exact for the whole circuit up to a global phase, because the other qubits are not touched, the same local state may occur with different states of the rest and the identity holds for each of them, and the set is the union over all reachable full states (we call these \emph{local windows} to distinguish them from the windows of condition (W)). The search is the beam search of the core on that reduced set of states, and every accepted circuit is certified as above. Windows that must be exact on all $2^k$ local states cannot remove a Toffoli that the context never fully uses. Windows with context can: on the ripple adder with three bits they lower the exact decomposition from $49$ to $29$ CNOTs (stage 1: $31$). For four and five bits they stay at or just above the stage-1 level ($41$ and $53$ against $41$ and $51$), so there the QCPT circuit is the stage-1 circuit.

\paragraph{Certificate.} A circuit found by the search is accepted only if the whole-circuit certificate of \S\ref{sec:method-verify} passes, dense for at most $12$ qubits. If it fails or the search finds nothing cheaper, QCPT returns the last certified circuit, so resynthesis cannot raise the two-qubit count or make the output wrong. We ran it up to twelve qubits (the incrementer with six bits); both searches are heuristic or restricted to a shape, so we do not claim optimality.

\subsection{Soundness}
\label{sec:soundness}

\paragraph{Semantics.}
In this section, we treat a circuit $C$ on $n$ qubits as a unitary $U_C$ followed by a terminal measurement of designated output qubits in the computational basis. The valid inputs are computational-basis strings $\mathcal X\subseteq\{0,1\}^n$, where the data qubits are free and the ancillae are set to $\ket0$. Note that $C$ may act on any state supported on $\mathcal H_{\mathrm{in}}=\mathrm{span}\{\ket x:x\in\mathcal X\}$, which includes coherent superpositions of data, and $P_{\mathrm{in}}$ denotes the projector onto $\mathcal H_{\mathrm{in}}$.

Let us consider two circuits $C,C'$ that act on the same qubits with the same measured outputs. We refer to them as \emph{observationally equivalent}, $C\equiv_{\mathrm{obs}}C'$, if they give the same output distribution for every state supported on $\mathcal H_{\mathrm{in}}$, and as \emph{subroutine equivalent}, $C\equiv_{\mathrm{sub}}C'$, if $U_{C'}P_{\mathrm{in}}=e^{i\theta}U_CP_{\mathrm{in}}$ for a single phase $\theta$. The former is the property needed by a stand-alone program, while the latter is needed by a calling circuit in which the output takes part in interference.

\noindent
It can be shown that subroutine equivalence implies observational equivalence. Moreover, it is preserved under composition, i.e., sequentially when the image $U_{C}P_{\mathrm{in}}$ lies in the input subspace of the next block, and in parallel with any environment, since $U_{C'}P_{\mathrm{in}}\otimes I_{\mathrm{env}}=e^{i\theta}U_CP_{\mathrm{in}}\otimes I_{\mathrm{env}}$ for any environment. Note that both notions quantify over superpositions, as a diagonal phase $D$ is a global scalar on every basis state and it becomes observable only if branches with different phases interfere. However, a block that is called under control, or controlled by a quantum register, needs $\theta=0$. Condition (C) gives $\theta=0$, while (R), (W) and the resynthesis can give $\theta\neq0$. For example, if all three qubits of $g$ are prepared in $\ket{101}$, the mirrored Margolus gate (phase $-1$ on $101$) is admitted by (R) with $\theta=\pi$, and a circuit that calls this block under control is then no longer equivalent. This is why we report $\theta$ for every output of QCPT, which allows a caller that needs $\theta=0$ to reject the others, and in the primary suite $\theta=0$ holds for every output.

\paragraph{Reachable support and admissibility.}
For a gate $g$ we write $U_C=W_gU_{\mathrm{CCX}}V_g$, where $V_g$ is the prefix and $W_g$ is the suffix. A set $S(g)\subseteq\{0,1\}^n$ is referred to as a \emph{reachable support} of $g$ if $V_g\mathcal H_{\mathrm{in}}\subseteq\Reach(g)=\mathrm{span}\{\ket x:x\in S(g)\}$ with projector $P_g$, i.e., it contains the basis strings that can be produced by the preceding gates. We obtain it by the forward propagation described in \S\ref{sec:method-context}. Note that coordinate subspaces of this kind are a special case of the assertions used in quantum Hoare logic (see, e.g., Refs.~\cite{ZhouYuYing2019,YuPalsberg2021}).

A decomposition $d$ of a Toffoli $g$ with unitary $\tilde U_d$ is \emph{admissible} if one of the following holds:
\begin{enumerate}
\item[(E)] \emph{exact}: $\tilde U_d=U_{\mathrm{CCX}}$,
\item[(R)] \emph{exact on the reachable subspace}: $\tilde U_dP_g=e^{i\theta}U_{\mathrm{CCX}}P_g$ for a single phase $\theta$,
\item[(U)] \emph{unobservable phase}: $\tilde U_d=DU_{\mathrm{CCX}}$ with $D$ diagonal and $W_gDW_g^\dagger$ diagonal.
\end{enumerate}
A pair $(g,g')$ of Toffolis on the same qubits with decompositions $(d,d')$ is admissible if
\begin{enumerate}
\item[(C)] \emph{canceling pair}: $\tilde U_d=DU_{\mathrm{CCX}}$, $\tilde U_{d'}=U_{\mathrm{CCX}}D^\dagger$ with $D$ diagonal commuting with every gate of the segment $\Sigma_{gg'}$ between them, or
\item[(W)] \emph{certified window}: the segment from $g$ to $g'$, with both substitutions and all pairs and windows admitted inside it, equals the exact segment as an operator up to one global phase. Admitted windows and pairs are nested or disjoint, and substitutions under (R) are made afterwards in the circuit with the window.
\end{enumerate}
In the \emph{bounded-approximate} variant of (R), we require $\min_\theta\lVert(\tilde U_d-e^{i\theta}U_{\mathrm{CCX}})P_g\rVert\le\epsilon$.

\noindent
Condition (R) is a single operator identity with one phase on $\Reach(g)$, and it cannot be replaced by the weaker requirement $\tilde U_d\ket x\propto U_{\mathrm{CCX}}\ket x$ imposed for each $x\in S(g)$ separately. To see this, let us consider the Qiskit relative-phase Toffoli gate RCCX, which equals $DU_{\mathrm{CCX}}$ with $D=-i$ on $\ket{110}$, $+i$ on $\ket{111}$ and $-1$ on $\ket{101}$ in the order $abt$, where the last state lies outside the support. We set both controls to $1$ and apply a Hadamard gate to the target before and after the gate. In this case, each reachable basis state acquires only a phase and a test performed state by state accepts RCCX. However, the phases differ by $\pi$, thus for the input $\ket{110}$ the exact circuit returns $t=0$ with probability $1$, while RCCX returns $t=1$. The Margolus gate used in our count-greedy baseline has a single phase $-1$ on $100$, which is not reached by this circuit when both controls are fixed to $1$, and the gate leaves the circuit unchanged. If the two controls are data inputs, the state $100$ is reachable and the circuit changes as well. Condition (R) rejects RCCX, since no single $\theta$ exists, and (U) rejects it as well, since $H$ does not map the phase to a diagonal operator. Condition (U) generalizes the removal of diagonal gates before a measurement, as it is done in the Qiskit pass \texttt{RemoveDiagonalGatesBeforeMeasure}. Condition (C) refers to the matched-phase structure of Barenco et al.~\cite{Barenco1995}, and we establish the required commutation by demanding that every gate between $g$ and $g'$ uses the qubits of $g$ at most as controls.

\begin{lemma}[support invariance]\label{lem:support}
Let $C_1$ be obtained from $C_0$ by one substitution admissible under (E), (R) or (C) with $\epsilon=0$. Then every reachable support $S(h)$ of a gate $h$ of $C_0$ which is not part of the substitution is also a reachable support of $h$ in $C_1$. The same holds under (W) for every gate outside the window. It does not hold in general for a gate inside a window or after a substitution under (U), see below.
\end{lemma}
\begin{theorem}[soundness]\label{thm:sound}
Let $C_0$ be the exact decomposition of a Toffoli network and let $C_k$ be obtained from $C_0$ by $k$ substitutions. Each substitution is admissible in the circuit produced by the previous ones under (E), (R), (C) or (W) with $\epsilon=0$, with supports of that circuit or their over-approximations. Then $C_k\equiv_{\mathrm{sub}}C_0$. If some of the steps are admissible under (U), then $C_k\equiv_{\mathrm{obs}}C_0$.
\end{theorem}
\noindent
The theorem requires the support of each step to be computed on the circuit produced by the previous steps. By Lemma~\ref{lem:support}, supports computed on $C_0$ stay valid across substitutions under (E), (R) and (C), so those selected on $C_0$ can be applied together, provided that the substitutions under (C) are made before those under (R) and that the structural rule of (C), that no gate of the segment writes a qubit of $g$, holds for the segment of $C_0$ (the substitutes of (R) then commute with $D$ as well). Inside a window and after a substitution under (U) the state entering a gate can differ from its state in $C_0$ by more than a phase, so a substitution under (R) there has to be checked against the support of the circuit that already contains the window or the (U) step. QCPT computes the supports of condition (R) on the input circuit $C_0$, which Lemma~\ref{lem:support} justifies for (E), (R) and (C) and verifies each pair under (C) against the exact circuit one by one and each window per segment, while the substitutions under (R) are not checked individually; the gates inside the segment of a pair or of a window are not excluded from (R). A gate inside a window is not covered, and neither is program mode with (U); there the guarantee comes from the whole-circuit certificate of \S\ref{sec:method-verify}, not from the stepwise argument, and QCPT returns the exact circuit if the certificate fails. The theorem therefore states what the local conditions guarantee for each admitted step in sequence, and the certificate makes the output independent of the order of the steps.

\paragraph{Proofs.} For the lemma, a gate $h$ placed before the substituted part has an unchanged prefix. Under (E) the circuit does not change. Under (R), a gate $h$ placed after $g$ receives $\tilde U_dV_g\ket\phi=\tilde U_dP_gV_g\ket\phi=e^{i\theta}U_{\mathrm{CCX}}V_g\ket\phi$, i.e., its state in $C_0$ up to a phase. Under (W) the same holds after the window, because the segment equals the exact segment up to a phase. Under (C) the diagonal $D$ commutes with every gate of the segment, so a gate between $g$ and $g'$ receives its state in $C_0$ multiplied by $D$, which lies in the same coordinate subspace, and a gate after $g'$ receives its state in $C_0$ exactly. In each case $S(h)$ remains valid.

The restrictions are necessary. Under (U), $D$ does not commute with the gates after $g$. Let $C_0=H_a\,\mathrm{CCX}(a,b,t)\,H_a\,h_2\,h_3\,H_a$ with $h_2=h_3=\mathrm{CCX}(a,c,d)$, $a$ initially $\ket0$, $b=0$ and $(c,d)=(1,0)$, and let (U) replace the Toffoli by $Z_a\mathrm{CCX}$. The suffix $W=H_ah_2h_3H_a=I$, so $WDW^\dagger$ is diagonal and the step is admissible, and both circuits give the same output distribution. The probability of $a=1$ at $h_2$ is $0$ in $C_0$ and $1$ in $C_1$, so the support $\{a=0\}$ of $C_0$ is not a support in $C_1$; a substitution under (R) that drops $h_2$ alone on this support flips $d$ in $C_1$ (dropping $h_3$ as well restores the output). QCPT does not select these drops on this circuit: the box propagation releases $a$ after the Hadamard gate, so its support at $h_2$ is not $\{a=0\}$, and $h_2$ and $h_3$ form a pair under (C), which it replaces by two relative-phase Toffolis; it also removes the first Toffoli, whose control $b$ is $\ket0$. The output is right (the repository selector gives $6$ CX for the $18$ of the exact circuit, with the whole-circuit certificate); the example breaks the stepwise argument, not the output of QCPT here. Inside a window the end substitutions need not be of the form $D\,U_{\mathrm{CCX}}$: the segment can equal the exact one as a whole while the state entering an interior gate differs from its state in $C_0$ by more than a phase. We built such an instance on five qubits (a window with a Toffoli at both ends, ends $Z_aU_{\mathrm{CCX}}$ and a diagonal correction, and an interior $H_a\,\mathrm{CCX}(c,d,a)\,H_a$), in which the support of the interior gate is $\{a=0\}$ in $C_0$ and $\{a=1\}$ in $C_1$, and the mirrored Margolus gate, exact on the support of $C_0$, then changes the output. With the Margolus gadgets of QCPT such a failure needs a non-monomial interior gate such as a Hadamard; in windows made of monomial gates alone, coordinate supports cannot change. Our supports are over-approximations by coordinate boxes, which a Hadamard turns into free qubits, so this does not occur for the box supports of \S\ref{sec:method-context}, but the lemma is not stated for them alone.

The theorem follows from the one-step statements given below by transitivity of $\equiv_{\mathrm{sub}}$ and $\equiv_{\mathrm{obs}}$, since $\equiv_{\mathrm{sub}}$ implies $\equiv_{\mathrm{obs}}$. Each admissible substitution preserves the respective equivalence, i.e., under (E), (R), (C) or (W) the new circuit is subroutine equivalent to the old one, and under (U) it is observationally equivalent to it. In the case of (R), let $\ket{\phi_g}=V_g\ket\phi\in\Reach(g)$. Then $\tilde U_d\ket{\phi_g}=\tilde U_dP_g\ket{\phi_g}=e^{i\theta}U_{\mathrm{CCX}}\ket{\phi_g}$, where $\theta$ does not depend on $\phi$, and by linearity we obtain $U_{C'}P_{\mathrm{in}}=e^{i\theta}U_CP_{\mathrm{in}}$. In the case of (C), the substituted segment reads $U_{\mathrm{CCX}}D^\dagger \Sigma_{gg'}DU_{\mathrm{CCX}}=U_{\mathrm{CCX}}\Sigma_{gg'}U_{\mathrm{CCX}}$, hence the unitaries are equal. For (W), the new segment equals the exact one up to a phase, and the substitutions already made inside it keep their segments up to phases by induction on the nesting, which gives $U_{C'}=e^{i\theta''}U_C$. For (U), we have $U_{C'}=D'U_C$, where $D'=W_gDW_g^\dagger$ is diagonal and commutes with the projectors onto the output strings. Note that it is enough to use a support $\widehat S(g)\supseteq S(g)$, as $\widehat P_gP_g=P_g$ and condition (R) on $\widehat S$ implies (R) on $S$.

A bounded-approximate step with tolerance $\epsilon$ changes the output distribution by at most $\epsilon$ in total variation, because the total-variation distance is bounded by the trace distance, which is at most $\lVert U_{C'}\ket\phi-e^{i\theta}U_C\ket\phi\rVert\le\epsilon$. By the triangle inequality, $m$ such steps, each certified in the circuit produced by the previous ones, change the distribution by at most $\sum_i\epsilon_i$. In the case of $\epsilon>0$, the final certificate is a norm-tolerance test instead of a test of exact equivalence, and our claim is limited to the above total-variation bound, while the equivalence statements below hold for $\epsilon=0$.

By transitivity, a sequence of admissible substitutions, each made in the current circuit with the support recomputed and possibly over-approximated, yields a circuit observationally equivalent to the exact decomposition of the input. If no step uses (U), the resulting circuit is also subroutine equivalent to it. The budget objective affects only the cost, since a candidate that fails its test is never accepted and the exact decomposition is always admissible.

\noindent
A phase admitted only under (U) may become observable when the circuit is used as a subroutine, which is how Toffoli networks are typically used. For this reason, we exclude (U) by default in QCPT and state all results of \S\ref{sec:eval} for subroutine equivalence (both semantics are shown in figure~\ref{fig:cost}a). By contrast, a pattern matcher that inserts relative-phase Toffolis wherever they fit produces an incorrect circuit each time the phase does not cancel (see \S\ref{sec:eval-safety}).
\subsection{Error accounting}
\label{sec:method-account}
Let $C'$ be a compiled circuit with approximation certificate $\epsilon$ (\S\ref{sec:soundness}), i.e., $\min_\theta\lVert(U_{C'}-e^{i\theta}U_C)P_{\mathrm{in}}\rVert\le\epsilon$, and let gate $e$ be implemented by a channel $\mathcal N_e$ with $\eta_e=\tfrac12\lVert\mathcal N_e-\mathcal U_e\rVert_\diamond$; for two-qubit depolarizing noise of average infidelity $\varepsilon_e$ one has $\eta_e=\tfrac54 \varepsilon_e$, so the budget $\sum_e \varepsilon_e$ of \S\ref{sec:method-select} is the hardware term up to this factor. For every input state supported on $\mathcal H_{\mathrm{in}}$ and every measurement outcome,
\begin{equation}
\bigl|p_{\mathrm{noisy}}(C')-p_{\mathrm{ideal}}(C)\bigr|\;\le\;\epsilon+\sum_e\eta_e .
\end{equation}
The trace distance between the noisy and the ideal output of $C'$ is at most $\sum_e\eta_e$, because the diamond norm is subadditive under composition (hybrid argument over the gates). The trace distance between the ideal outputs of $C'$ and $C$ is at most $\lVert U_{C'}\ket\phi-e^{i\theta}U_C\ket\phi\rVert\le\epsilon$, and the probability of an outcome changes by at most the trace distance. Hence every outcome probability of the compiled circuit lies within the bound of the ideal exact circuit, provided that the two-qubit gates are the only source of error and that their noise is depolarizing. The bound is pessimistic for approximations: the dropped phases of the Fourier adder give $\epsilon\approx0.2$ per gate against $\eta_e\sim10^{-3}$, so a rule based on it would admit none. When the exact circuit outputs a basis string on the measured register (adders, oracles), the approximation enters quadratically, $p_{\mathrm{ideal}}(C')\ge1-\epsilon^2$ and so $p_{\mathrm{noisy}}(C')\ge1-\epsilon^2-\sum_e\eta_e$. This sum is a conservative criterion that QCPT does not apply as such. For the Fourier adder we replace $\epsilon^2$ by the exactly simulated loss of the ideal success probability and rank the cutoffs by the calibrated cost, which gives the choices of \S\ref{sec:eval-approx}. The bound counts only two-qubit gates, with $\eta_e=\tfrac54\varepsilon_e$ valid for depolarizing noise alone; single-qubit, idle and readout errors lie outside it, and a coherent error of the same average infidelity has a larger $\eta_e$ (a coherent rotation with $\varepsilon_e=1.4\times10^{-3}$ gives $\eta_e\approx0.04$). On hardware (Table~\ref{tab:approx}, $247$ and $268$ routed CZ for the exact $6$-bit adder in the two compilations, $\varepsilon_e=1.4\times10^{-3}$) the bound with the calibrated $\eta_e=\tfrac54\varepsilon_e$ would give $p\ge1-247\cdot\tfrac54\cdot1.4\times10^{-3}=0.57$ and $p\ge1-268\cdot\tfrac54\cdot1.4\times10^{-3}=0.53$, against mean successes of $0.55$ and $0.54$, but $4$ of the $6$ inputs in each compilation lie below it. We read this as a violation of the depolarizing, two-qubit-only model, not as a measurement of the underestimate of $\varepsilon_e$.

\subsection{Cost on a reachable support}
\label{sec:method-cost}
Let $L(g)\subseteq\{0,1\}^3$ be the projection of $S(g)$ onto the
bits $abt$ of $g$, $\Pi(g)$ the projection of $L(g)$ onto the controls $ab$,
and $Q=\{100,101,010,011\}$. For a circuit on the qubits of $g$ only,
(R) on $\Reach(g)$ is the same identity on the span of $L(g)$.

\begin{proposition}[Cost of a Toffoli on its reachable support]\label{prop:cost}
Let $c^*(g)$ be the smallest number of CNOT gates in a circuit on the qubits
of $g$, with arbitrary single-qubit gates, that satisfies condition (R) for
$g$. Then, with $L=L(g)$ and $\Pi=\Pi(g)$,
\[
c^*(g)=\begin{cases}
0 & 11\notin\Pi \text{ or } \Pi=\{11\},\\
1 & \text{else, } 01\notin\Pi \text{ or } 10\notin\Pi,\\
2 & \Pi=\{01,10,11\},\\
3 & \Pi=\{0,1\}^2,\ Q\not\subseteq L,\\
4 & \Pi=\{0,1\}^2,\ Q\subseteq L\neq\{0,1\}^3,\\
6 & L=\{0,1\}^3.
\end{cases}
\]
Every value is attained by an explicit circuit, and the upper bounds are explicit circuits and the lower bounds are sketched in Appendix~\ref{app:cost} (full proofs and an exhaustive check over all $255$ supports are in the repository); the bound $6$ is the theorem of Shende and
Markov~\cite{ShendeMarkov2009}.
\end{proposition}

\noindent
On any three of the four control values $a\wedge b$ is affine, e.g.,
$a\wedge b=a\oplus b\oplus 1$ without $00$, and the Toffoli needs at most two
CNOTs. A nonlinear gate is needed only when all four control values reach
$g$, and a missing string of $Q$ then allows a relative-phase Toffoli. QCPT
uses the classes $0$, $1$ and $3$ (the latter when $100$ or $101$ is missing)
and Margolus gates in pairs and windows, so its two-qubit count can be predicted
statically. The prediction coincides with the outputs of the released code for all $528$
Toffolis of the $82$ circuits of the primary, RevLib (at most $12$ lines) and scale suites. The classes not yet used, i.e., the swapped orientation of the Margolus gate when $010$ or $011$ is missing, would
remove a further $3$ two-qubit gates on the primary suite, $51$ on the $50$ RevLib circuits and none on the scale suite.

Figure~\ref{fig:method} shows an admitted and a rejected instance.

\begin{figure*}[t]
\centering
\begin{subfigure}[t]{0.47\textwidth}
\centering
\begin{quantikz}[column sep=9pt,row sep={16pt,between origins}]
\setwiretype{n} & \push{g_1\!:\,D} & & \push{g_2\!:\,D^{-1}} & \\[-4pt]
\lstick{$q_0$} & \ctrl{2} & \qw & \ctrl{2} & \qw \\
\lstick{$q_1$} & \ctrl{1} & \qw      & \ctrl{1} & \qw \\
\lstick{$a{=}\ket0$} & \targ{} & \ctrl{1} & \targ{} & \qw \\
\lstick{$t$} & \qw      & \targ{}  & \qw      & \meter{}
\end{quantikz}
\caption{\textbf{Admitted.} The phase $D$ of the substituted gate at $g_1$ is undone by $D^{-1}$ at $g_2$ before anything can interfere (condition~(C)); two-qubit count $13\to7$.}
\end{subfigure}\hfill
\begin{subfigure}[t]{0.47\textwidth}
\centering
\begin{quantikz}[column sep=9pt,row sep={16pt,between origins}]
\setwiretype{n} & & \push{g\!:\,D} & & & \\[-4pt]
\lstick{$q_0$}       & \qw & \ctrl{2} & \qw & \gate{H} & \meter{} \\
\lstick{$q_1$}       & \qw & \ctrl{1} & \qw & \qw      & \qw      \\
\lstick{$t$}         & \qw & \targ{}  & \qw & \qw      & \qw
\end{quantikz}
\caption{\textbf{Rejected.} No uncompute follows $g$, and the Hadamard makes the sign of the branch $100$ interfere: for $q_0=\ket+$, $q_1t=00$ the exact circuit gives $q_0=0$ and the substituted circuit $q_0=1$, so (R) and (C) fail.}
\end{subfigure}
\caption{The per-instance check decides the admissibility of a relative-phase Toffoli per instance. A pattern matcher applies the cheap variant in both cases and corrupts circuit~(b).}
\label{fig:method}
\end{figure*}

%=======================================================================
\section{Evaluation}
\label{sec:eval}
%=======================================================================
We investigate (Q1)~whether the per-instance check prevents the undetected
errors that count-greedy or template substitution introduces,
(Q2)~whether budget-aware selection reduces two-qubit infidelity compared with
count-based and exact-only baselines, (Q3)~how far per-instance
verification scales, (Q4)~whether the gain survives on hardware, and (Q5)~how QCPT compares with an iterative search by a language model. Unless stated otherwise, the results of \S\ref{sec:eval-safety}--\S\ref{sec:eval-scal} are those of stage 1; \S\ref{sec:eval-resynth} adds the resynthesis, and \S\ref{sec:eval-hw} tests both in simulation and on hardware.

\subsection{Setup}
\textbf{Benchmarks.} The primary suite has $12$ circuits of at most $8$ qubits (Appendix~\ref{app:suite}) sized for exhaustive verification, with the structure of each Toffoli (compute/uncompute, live, mixed) known by construction. The scale suite has $20$ reversible-arithmetic and oracle circuits at $12$--$24$ qubits (Grover AND ladders, carry-lookahead, ripple and live-carry adders, modular incrementers, array multipliers) certified with decision diagrams; we report its non-improving families ($0\%$). We check every circuit that computes a named function on all basis inputs. The structural models do not implement the full arithmetic, which does not affect soundness, since every output is certified against the exact decomposition of the same circuit. The RevLib and state-resetting circuits come from independent work, but the $20$ small RevLib circuits of the resynthesis suite are in-sample: we used them while developing stage 2 and did not hold them out, so their saving is not a test on unseen circuits. We provide the generators in the repository.

\textbf{Baselines.} We compare with tket~\cite{tket2020}, Qiskit~\cite{Qiskit2024} at optimization level~3 (opt-3), an exact-only selector and a count-greedy selector. The count-greedy selector substitutes the three-CNOT Margolus gate at every CCX and so bounds unchecked savings from above; replacing it by Qiskit's RCCX, which has three residual phases instead of one, corrupts the same $6$ of $12$ circuits. QContext uses only exact variants and Maslov's identities apply only where the phase provably cancels, so neither can serve as such a baseline.

\textbf{Metrics.} We report the estimated circuit infidelity (the error budget), the two-qubit count, the depth and the ancilla width, with $p_{2q}=10^{-2}$ and $p_{1q}=10^{-3}$ unless stated otherwise. All results are for subroutine equivalence on the input subspace with clean ancillae in $\ket0$, and every output, baselines included, passes the same whole-circuit certificate (dense up to $12$ qubits, QCEC above; \S\ref{sec:method-verify}).

\subsection{Safety ablation}
\label{sec:eval-safety}
The ablation of figure~\ref{fig:safety} separates the two layers of QCPT on the primary suite: (a)~count-greedy relative-phase substitution at every Toffoli, (b)~substitution driven by our context analysis including the program-only condition (U) but without the check, and (c)~the same substitution gated by a per-instance check of whole-circuit equivalence.

\begin{figure}[t]
\centering
\includegraphics[width=0.92\columnwidth]{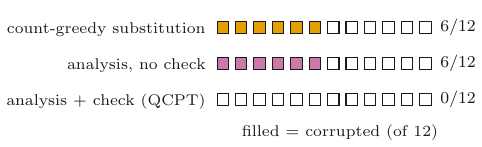}
\caption{Safety ablation: circuits corrupted across the
primary suite (lower is better). Count-greedy relative-phase
substitution and our analysis
with condition (U) but without the gate each break $6$ of $12$
circuits, and only QCPT with the check is certified error-free ($0$/12). Colors in all figures: vermillion QCPT,
blue exact-only, yellow count-greedy.}
\label{fig:safety}
\end{figure}

Variants (a) and (b) each corrupt $6$ of $12$ circuits, and none of the errors is reported. In (b) five of the six errors come from condition (U), and the sixth, \texttt{half\_uncomputed}, is rejected also in program mode: the forward-cone test flags a live Toffoli when no later gate lets its phase interfere before the final measurement, which is valid for a complete program but not for a subroutine whose caller interferes the outputs. In (c), the check in subroutine mode, the check rejects every such substitution and QCPT uses the exact decomposition. Restricted to (C), (W) and (R), the admission rules change none of the $12$ circuits (all outputs are certified), as the soundness result (\S\ref{sec:soundness}) predicts; the check also guards against implementation errors and rejects no valid substitution.

\paragraph{Two semantics.} As subroutines, count-greedy ($167$ two-qubit gates) changes $6$ of $12$ circuits, stage 1 in program mode ($167$) changes $5$, and stage 1 in subroutine mode ($194$) none (figure~\ref{fig:cost}a). As complete programs measured at the end, all methods are correct on this suite and count-greedy is cheapest; verification then matters only where a phase can reach an interfering gate, as in the $3$ interference circuits, where count-greedy changes the measured outcome of all $3$ and QCPT of none.

\subsection{Error-budget reduction}
\label{sec:eval-budget}
Relative to exact-only selection ($281$ two-qubit gates, infidelity $2.74$), stage 1 with the check attains $194$ / $1.97$ ($-31.0\%$ / $-28.1\%$), $-26.7\%$ in infidelity relative to Qiskit opt-3 and $-25.9\%$ relative to tket, and the total depth falls from $477$ to $362$ (figure~\ref{fig:budget}). Count-greedy substitution is cheaper only because it is not correct as a subroutine.

\begin{figure}[t]
\centering
\includegraphics[width=0.92\columnwidth]{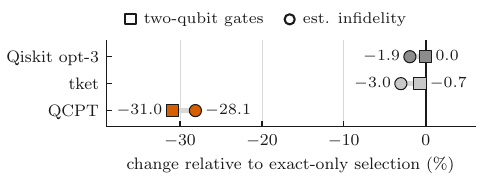}
\caption{Error budget over the primary suite relative to exact-only selection ($281$
two-qubit gates, summed infidelity $2.74$; lower is better, no ancilla overhead for any
method). Squares: summed two-qubit count, circles: summed estimated infidelity (two- and
single-qubit gate errors). All four methods pass the certificate on every circuit.}
\label{fig:budget}
\end{figure}

\begin{figure*}[t]
\centering
\includegraphics[width=0.92\textwidth]{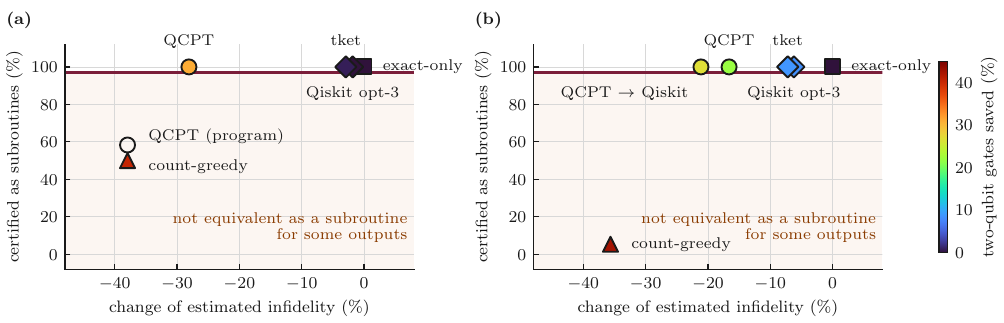}
\caption{Cost and correctness of the methods. Change of the estimated infidelity relative to
exact-only selection against the share of outputs that pass the certificate as
subroutines (higher is safer) on (a)~the primary suite and (b)~the $56$ RevLib circuits. The
color gives the two-qubit gates saved relative to exact-only selection. In
(a) the open marker is stage 1 in program mode: it removes as many gates as count-greedy
substitution ($167$ against $281$), changes $5$ of $12$ outputs as subroutines and none as a
complete program, which holds for all methods. On the three interference circuits count-greedy
substitution changes all $3$ outputs in both semantics, and QCPT none. Qiskit and tket
outputs pass the same certificate on every circuit. Constant-$0$ lines of the RevLib
circuits are pinned as clean ancillae.}
\label{fig:cost}
\end{figure*}

\paragraph{Gain attribution.} The $10$ structural pairs of condition (C) give $69\%$ of the saving ($281\to221$), the $4$ certified windows of (W) $28\%$ ($\to197$) and one standalone mirrored Margolus gate under (R) $3\%$ ($\to194$). Windows occur in the adders, where the CNOTs between a Toffoli and its mirror image write a control and are undone symmetrically, so the structural rule rejects the pair although the phases cancel. (R) is rare here because the live Toffolis write into wires not declared as ancillae, and no control is constant on the reachable subspace when all data inputs are free. In a density-matrix simulation with depolarizing noise ($p_{2q}{=}10^{-2}$) QCPT with the check raises the fidelity of the four-control multi-controlled-X (MCX-$4$) oracle from $0.74$ to $0.86$ and of the controlled adder from $0.68$ to $0.78$, in the order predicted by the budget.

\subsection{Mirror tests in simulation and on hardware}
\label{sec:eval-hw}
Here, we test subroutine equivalence directly with a mirror circuit, since running the circuits as complete programs cannot reveal a wrong relative phase. For an exact decomposition $E$ and a candidate $V$ we run $H_dV^\dagger EH_d$ on $\ket{0\cdots0}$, where $H_d$ puts the data qubits in uniform superposition. If $V$ is subroutine equivalent to $E$, the circuit returns to $\ket{0\cdots0}$ with probability $P=1$, and a relative phase on the reachable states lowers this probability. We compare two compilations of the same mirror circuit for \texttt{ibm\_aachen}, both of which are transpiled to the device before they are run (instruction-set-architecture circuits). In the case of \emph{Qiskit}, the exact decomposition of the original circuit is transpiled with optimization levels $2$ and $3$ and ten seeds each ($20$ transpilations of the whole mirror circuit), and we keep the best result. In the case of \emph{QCPT}, we take the best of four candidates by predicted success probability. These candidates include QCPT before routing, with the result routed by Qiskit, and QCPT with its own router, which respects the coupling graph and avoids triangles of qubits on the heavy-hexagonal lattice, as a portfolio with the route of the SABRE algorithm~\cite{SABRE2019} (best of $200$ seeds) applied to the circuit of stage~2 and to the circuit of stage~1, where the route with fewer two-qubit gates wins. The routing budgets are therefore unequal: the SABRE candidate has $200$ seeds on each of two circuits, against ten seeds per level for the whole Qiskit circuit, and the half $E$ of its mirror circuit is transpiled once, at level $1$, where the Qiskit candidate searches levels $2$ and $3$ and ten seeds. This candidate won $4$ of the $27$ hardware points (Appendix~\ref{app:bench}): its route was QCPT's own router once and SABRE three times, on the circuit of stage~1 twice and on the circuit of stage~2 once (Appendix~\ref{app:bench}). Moreover, we apply QCPT to the circuit that Qiskit has already placed and routed, where the whole $V^\dagger$ segment is resynthesized on the valid subspace, the ancillae are $\ket{0}$ at the output and the coupling-graph constraint acts as a filter of the search. The Qiskit circuit itself serves as the fourth candidate (a fall-back), and by construction the predicted success probability of the QCPT circuit cannot be lower than that of Qiskit. Note that we verified that each Qiskit circuit is equal to an independent Qiskit-only transpilation of the original circuit, i.e., it is not modified by QCPT. In all runs reported in this subsection dynamical decoupling, Pauli twirling and error mitigation were disabled.

\paragraph{Simulation, $33$ circuits.} We simulated both compilations of all $33$ circuits of the resynthesis suite (Section~\ref{sec:eval-resynth}) with a density-matrix simulation, in which the noise model, referred to as the \emph{calibration model}, is built from the calibration of \texttt{ibm\_aachen} of 29 October 2026, 03:20, which is also the calibration of the corrected noise model below. The results of this simulation (and not of a measurement) are shown in figure~\ref{fig:sim}. We find that QCPT returns more often than Qiskit on $29$ circuits by more than $0.01$ and on none of them less often by more than $0.01$. The mean gain is $0.030$ and the largest gain, $0.069$, is obtained for \texttt{controlled\_adder\_2b}. QCPT selects its circuit by the success probability predicted by the same calibration, so the simulation does not test this choice independently; only the measurement does. The final circuit is QCPT before routing in $19$ cases, after routing in $12$ and with its own router in $2$, and the Qiskit fall-back is never needed. In total, QCPT uses $2545$ two-qubit gates as compared to $3234$ for Qiskit ($-21\%$), and the gain grows with the number of gates saved by about $1.4\times10^{-3}$ per gate, as shown in figure~\ref{fig:sim}b.

\begin{figure*}[t]
\centering
\includegraphics[width=0.92\textwidth]{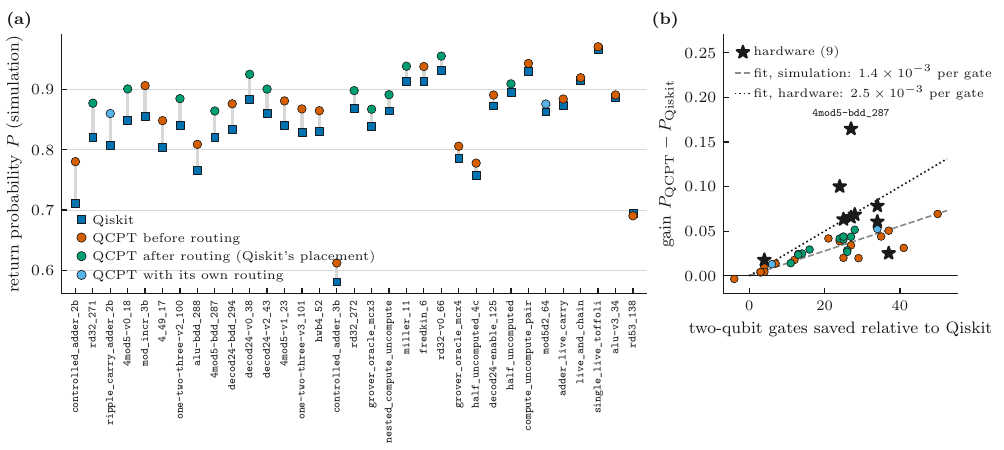}
\caption{\textbf{Simulated mirror tests of $33$ circuits on \texttt{ibm\_aachen}} (noise model from the calibration). \textbf{(a)}~Return probability $P$ of the Qiskit circuit (blue squares) and of QCPT (circles colored by the candidate that won), sorted by gain. \textbf{(b)}~Gain $P_{\mathrm{QCPT}}-P_{\mathrm{Qiskit}}$ as a function of the number of two-qubit gates saved relative to Qiskit. Stars denote the nine circuits measured on hardware (figure~\ref{fig:hw}) and lines are fits through the origin.}
\label{fig:sim}
\end{figure*}

\paragraph{Probe and regions.} The calibration does not show a failed coupler, so the region has to be found before the measurement, not after it. On 29 October we therefore measured every qubit and every coupler of \texttt{ibm\_aachen} with four circuits ($3$\,s of QPU time): $40$ X gates on all qubits, and $\mathrm{H}\,(\mathrm{CZ}\,\mathrm{CZ})^{15}\,\mathrm{H}$, an identity, on all $173$ couplers in three matchings (figure~\ref{fig:probe}). The coupler $133$--$134$ had an error of $8.3\%$ per CZ against $0.08\%$ in the calibration of the same day, and $31$ couplers were above $2\%$. We dropped the couplers above $1.2\%$ per CZ and the qubits with a readout flip above $0.2$ (the probe values are inflated by the parallel operation: median $0.5\%$ per CZ against $0.14\%$ in the calibration) and chose the best connected regions among the rest, with the coupler cost taken as the larger of the calibration and the probe divided by $3.5$. Every circuit was compiled for the smallest of these regions that holds it ($6$, $8$, $10$ or $12$ qubits), and we repeated the probe before every job; the failure was unchanged over the hour of the measurements.

\paragraph{Hardware, nine circuits.} We measured nine of the circuits on \texttt{ibm\_aachen} with $8000$ shots each ($18$ circuits in one job) on these regions. The Qiskit and QCPT circuits of the same example were run on the same qubits, as shown in figure~\ref{fig:hw}. We find that QCPT returns more often than Qiskit on all $9$ circuits, on each of them by more than $3$ standard errors (the smallest gains are $+0.018$ for \texttt{adder\_live\_carry}, $z=3.0$, and $+0.025$ for \texttt{mod\_incr\_3b}, $z=3.9$), so that all nine remain significant after a Bonferroni correction for nine tests. The mean gain is $0.072$ as compared to $0.043$ simulated for the same qubits and circuits, so the model underestimates the advantage by a factor of $1.7$; without \texttt{4mod5-bdd\_287} the mean gain is $0.060$. These statistics count only shot noise, with one run per circuit: the control \texttt{rl\_4mod5-v0\_18}, which we measured in two jobs an hour apart on the same qubits, gave $0.803\to0.861$ and $0.762\to0.830$ (the second is shown), so the day-to-day uncertainty is larger than the shot noise. The largest gain, $0.456\to0.621$ ($+0.165$, $21$ standard errors), is observed for \texttt{4mod5-bdd\_287}, where QCPT saves $27$ of $104$ two-qubit gates. Over the $18$ circuit instances (nine circuits, two compilations each), the simulation overestimates the absolute return probability by $0.11$ on average (Pearson correlation $0.65$), and the correlation between the measured and simulated gains over the nine circuits, which is equal to $0.28$, is not significant. Post-selecting the ancillae on $\ket{0}$ raises both probabilities (for \texttt{4mod5-bdd\_287} we obtain $0.601\to0.785$), and QCPT keeps more shots than Qiskit ($0.92$ against $0.90$ of the runs). Five further circuits that we fixed in advance, \texttt{ripple\_adder\_4b} and four RevLib circuits not measured before, were measured in an earlier job (table~\ref{tab:probed} in Appendix~\ref{app:bench}); QCPT returns more often on all five, by $0.034$ to $0.097$ ($5.0$ to $14.6$ standard errors).

\begin{figure*}[t]
\centering
\includegraphics[width=0.92\textwidth]{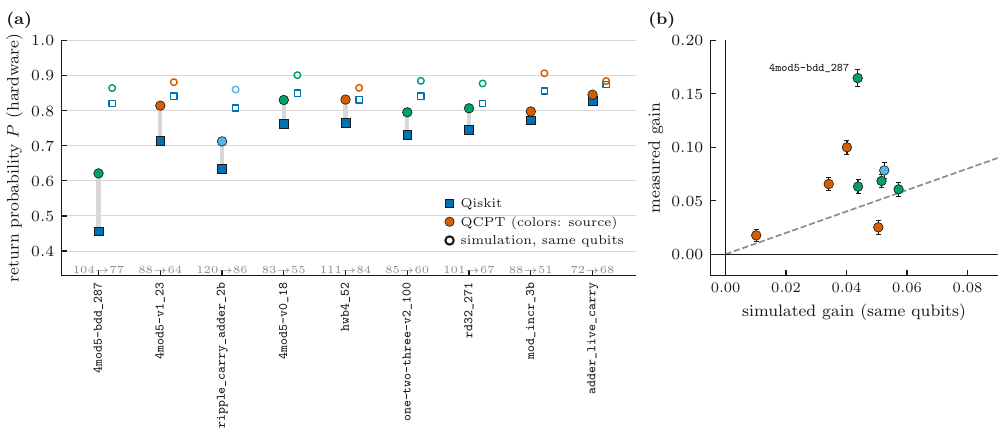}
\caption{\textbf{Hardware test on \texttt{ibm\_aachen}, probed regions.} \textbf{(a)}~Measured return probability $P$ ($8000$ shots, bars denote one standard error) of the Qiskit circuit (blue squares) and of QCPT (circles colored by the winning candidate as in figure~\ref{fig:sim}). Open markers show the simulation on the same qubits and labels give the numbers of two-qubit gates, Qiskit$\to$QCPT. \textbf{(b)}~Measured gain as a function of the simulated gain.}
\label{fig:hw}
\end{figure*}

\paragraph{Calibration against measurement.} We checked the calibration on the day of the measurement. On eight couplers of the regions used, benchmarked simultaneously, randomized benchmarking gives two-qubit gate errors of $1.3$--$3.0\times10^{-3}$, i.e., $1.4$--$3.3$ times the calibration values $0.72$--$1.67\times10^{-3}$ (median $2.0$; the pair $113$--$119$ is the worst), single-qubit errors of $1.1$--$5.8\times10^{-4}$ (median $1.8\times10^{-4}$, $1.6$ times the calibration) and the median readout error is $0.41\%$ against $0.44\%$. On $14$ isolated qubits $T_1$ (median $284\,\mu$s) and the echo $T_2$ with idle neighbors (median $421\,\mu$s) are $1.09$ and $1.05$ times the calibration, but the echo $T_2$ with the neighbors in superposition, the regime of the mirror circuits, has a median of $76\,\mu$s (range $49$--$395\,\mu$s), $4.8$ times shorter than with idle neighbors. A model built from the calibration alone therefore underestimates the noise of the mirror circuits.

\paragraph{Dynamometer view.} We interpret the gain through a dynamometer, which gives the name used here. Let the load $x$ be the number of two-qubit gates of the Qiskit circuit after transpilation. The return probability $P(x)$ corresponds to the torque and $xP(x)$, the expected number of correctly executed gates, to the power. Removing gates does not make the device less noisy, just as tuning an engine does not enlarge it; it moves the maximum of the power curve to a higher load. We fit the stretched exponential $P=\exp[-(x/x_0)^p]$ to the conditional return probability (ancillae post-selected on $\ket{0}$), separately for four families of four circuits each ($4$ to $10$ qubits: modular incrementer, ripple adder, Grover AND ladder, half-uncomputed oracle), with one exponent $p$ per family and separate $x_0$ for Qiskit and for QCPT. We measured exactly these circuits on \texttt{ibm\_aachen} on the probed regions ($15$ of them in one job at $8000$ shots, the ripple adder with $4$ bits in a second job), and the simulation uses the same transpiled circuits, so the loads coincide (figure~\ref{fig:dyno}).

On hardware, QCPT returns more often than Qiskit on all $16$ circuits of the four families, on $15$ of them by more than three standard errors (largest gains: \texttt{half\_uncomputed\_4c} $0.434\to0.572$, \texttt{mod\_incr\_4b} $0.555\to0.676$, \texttt{grover\_oracle\_mcx4} $0.546\to0.639$, ripple adder with $3$ bits $0.293\to0.381$); the exception is \texttt{grover\_oracle\_5c} ($0.244\to0.252$). For the ripple adder with $4$ bits QCPT returns with $0.176$ against $0.141$. The load scale grows by a factor of $1.49$ for the incrementer ($x_0=241\to361$, $p=1.57$), $1.36$ for the half-uncomputed oracle, $1.17$ for the ripple adder and $1.03$ for the Grover ladder. The corrected noise model defined below, fitted to the same $32$ measurements and therefore not a validation, gives $1.75$, $1.31$, $1.11$ and $1.22$ in the same order of families (incrementer, half-uncomputed oracle, ripple adder, Grover ladder), and the model from the calibration alone gives $2.1$, $1.6$, $1.0$ and $1.3$. The factors give only the order of magnitude: each fit has three parameters for eight values, the measurements and the corrected simulation order the families differently, and for $p=1$ the hardware ratios are $1.9$, $1.6$, $1.3$ and $1.05$. Leaving out one circuit gives $1.3$--$1.6$ for the incrementer, $1.2$--$1.4$ for the half-uncomputed oracle, $1.1$--$1.2$ for the ripple adder and $1.0$--$1.25$ for the Grover ladder, whereas the shot noise alone gives $95\%$ intervals of width at most $0.07$, so the uncertainty comes from the form of the fit and the choice of circuits, not from the number of shots.

\begin{figure*}[t]
\centering
\includegraphics[width=0.96\textwidth]{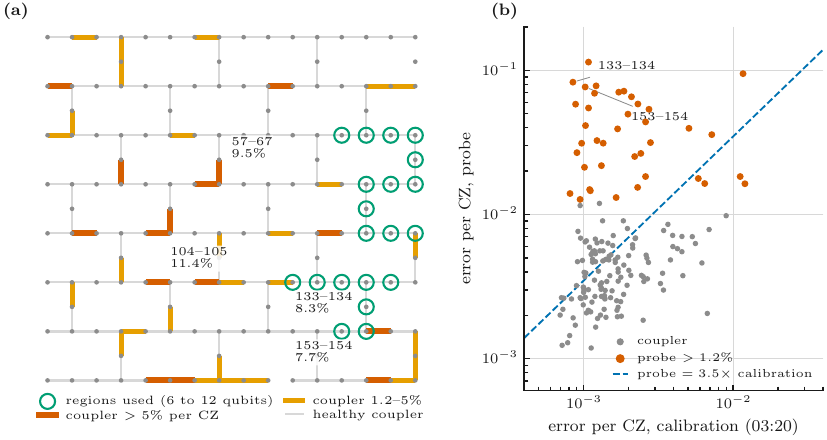}
\caption{Failed couplers of \texttt{ibm\_aachen} found by the probe of 29 October 2026 ($173$ couplers, $3$\,s of QPU time). (a)~Map of the $156$-qubit heavy-hex chip (the geometry is reconstructed from the coupling graph; gray dots are qubits). The colour of a coupler gives its error per CZ in the probe: yellow for $1.2$--$5\%$, orange above $5\%$. Circles mark the qubits of the regions used for the measurements ($6$, $8$, $10$ and $12$ qubits), which contain no coupler above $1\%$ per CZ in any of the probes of the measurements. (b)~Calibration of the same day (03:20) against the probe for every coupler. The dashed line is the median ratio $3.5$ (the probe drives all couplers in parallel and includes the Hadamard gates and the readout); the coupler $133$--$134$ lies about thirty times above it, and the failed couplers are not visible in the calibration.}
\label{fig:probe}
\end{figure*}

\paragraph{Failure of the calibration model and a corrected model.} On the $32$ measured compiled circuits ($16$ circuits, each compiled by Qiskit and by QCPT), the calibration model overestimates the (unconditional) return probability by $0.19$ on average (root-mean-square error $0.24$), more for the $8$ large circuits ($0.39$) than for the $24$ small ones ($0.12$), although it ranks the circuits well (Pearson correlation $0.85$). It has three weaknesses, two measured (Appendix~\ref{app:bench}) and one inferred: the two-qubit error is larger than in the calibration ($1.4$ to $3.3$ times), the echo $T_2$ with active neighbors is $4.8$ times shorter, and residual coherent errors, e.g., a static $ZZ$ phase, add up in amplitude, which is consistent with the exponents $p>1$ of the hardware fits. We add one correction for each and fit their three sizes to the $32$ measurements: the error of every CZ gate is multiplied by $K_2=1.94$, every $T_2$ by $R_2=0.073$ ($14$ times shorter, more than the measured factor, so that the fit absorbs a further effect that grows with the size of the circuit) and every CZ is followed by a coherent $ZZ$ rotation of $\theta=0.019$\,rad. The corrected model reduces the mean error from $+0.19$ to $-0.01$ and the root-mean-square error from $0.24$ to $0.091$ (in sample, filled markers in figure~\ref{fig:dyno}); fitted on the $24$ small circuits alone it lowers the root-mean-square error on the $8$ large ones, which were not used in the fit, from $0.41$ to $0.18$ (bias $+0.39\to+0.12$), so we regard it as an improvement that is still incomplete. Both models predict the sign of the gain of QCPT correctly for all $16$ pairs. The calibration model underestimates the mean gain ($0.041$ against $0.057$ measured) and the corrected model overestimates it ($0.079$), but neither tracks the gains of single circuits (correlations $-0.34$ and $-0.39$), so we do not use them to predict where QCPT gains most. For the conditional return probability, which the fits of the dynamometer use, the calibration model overestimates by $0.19$ (root-mean-square error $0.25$, Pearson $0.76$) and the corrected model by $0.03$ ($0.11$, Pearson $0.91$), and the mean gain is $0.067$ measured, $0.017$ and $0.055$ for the two models (correlations $0.44$ and $0.61$).

\begin{figure*}[t]
\centering
\includegraphics[width=0.92\textwidth]{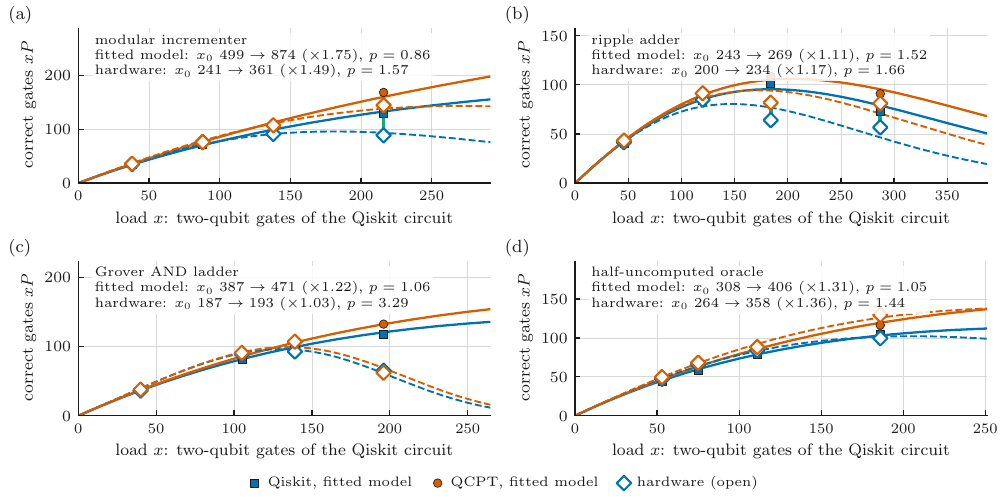}
\caption{Gain of QCPT as a dynamometer curve for four circuit families, at the same load for simulation and experiment. The expected number of correct gates $xP$ is shown as a function of the load $x$, i.e., the number of two-qubit gates of the transpiled Qiskit circuit (Qiskit in blue, QCPT in vermillion). Vertical segments join the two compilations of one measured circuit (green: gain). Open diamonds denote measurements on \texttt{ibm\_aachen} ($8000$ shots) with the dashed fit $x\exp[-(x/x_0)^p]$. All circuits were measured on the probed regions. Filled markers denote the simulation of the same circuits with the noise model with fitted corrections (see text) together with the solid fit of the same form. The return probability is conditional on the ancillae being $\ket{0}$. \textbf{(a)}~Modular incrementer. \textbf{(b)}~Ripple adder, the family with the weakest fits: its two largest circuits lie below the curves. \textbf{(c)}~Grover AND ladder. \textbf{(d)}~Half-uncomputed oracle.}
\label{fig:dyno}
\end{figure*}
\subsection{Workload, structure and independent benchmarks}
\label{sec:eval-struct}
\paragraph{Scale and structure.} On the $20$-circuit scale suite (\S\ref{sec:eval-scal}) stage 1 reduces the two-qubit count by $30.1\%$ and the infidelity by $22.2\%$ relative to exact-only selection, and the count by $28.5\%$ relative to Qiskit opt-3, using $53$ structural pairs, $17$ certified windows and $64$ mirrored Margolus gates (figure~\ref{fig:scale}a). The reduction is $49.2$--$49.6\%$ for Grover AND-ladder oracles and $44.8$--$46.2\%$ for their half-uncomputed variants (pairs), $49.4$--$49.7\%$ for carry-lookahead adders and $22.4$--$37.0\%$ for ripple adders (pairs, windows and mirrored Margolus gates), and $45.5$--$45.8\%$ for modular incrementers, which contain no pair, so mirrored Margolus gates alone give it. Array multipliers and live-carry adders give $0\%$: their Toffolis accumulate into registers not declared as ancillae, so $L(g)=\{0,1\}^3$, no pair or (R) applies, and no decomposition on their own qubits is cheaper than the exact one (Proposition~\ref{prop:cost}). The gain is therefore set by the fraction of Toffolis that are uncomputed or write into clean ancillae, not by the qubit count; the ripple adders are an exception because the number of Toffolis grows with the width while QCPT admits one pair, four windows and at most two mirrored Margolus gates, so the gain falls from $37.0\%$ to $22.4\%$. Routing and the error rates of five platforms ($p_{2q}/p_{1q}$ from $7$ to $47$) leave the decisions unchanged, and calibration-aware orientation changes the routed infidelity by $0.26\%$ on the $113$ instances where it reorients a gate (not significant, Wilcoxon $p=0.054$) (Appendix~\ref{app:orient}).

\paragraph{RevLib benchmarks.} On circuits that we did not write, we take every NOT/CNOT/Toffoli realization of RevLib~\cite{RevLib2008} with at most $24$ lines, pin the declared constant-$0$ lines and keep the garbage outputs in the certificate. Figure~\ref{fig:cost}b shows the $56$ circuits that finish within one hour (the other $6$ have more than $5000$ gates: two timed out and four were not attempted). QCPT removes $21.9\%$ of the two-qubit gates relative to exact-only selection and $14.9\%$ and $14.6\%$ relative to Qiskit opt-3 and tket, and every output is certified. The gain comes almost entirely from condition (R) ($127$ mirrored Margolus gates against $4$ windows and no structural pair), because RevLib realizations write intermediate results into constant-$0$ lines and rarely uncompute them. When Qiskit opt-3 runs after QCPT, QCPT removes $20.2\%$ relative to Qiskit alone; count-greedy gives the lowest counts but passes the certificate for only $3$ circuits.

\paragraph{State-resetting circuits.} Each letter of the quantum state-resetting protocol~\cite{QuantumResetting2025} applies $T_j^\dagger ST_j$, whose multi-controlled gates in CCX form write a predicate into a clean ancilla and uncompute it. QCPT removes $24.5$--$27.9\%$ of the native two-qubit gates across $8$-state words of length $3$ to $7$ and one $16$-state word, up to $12$ qubits, and certifies every output. 

\subsection{Budgeted approximation and calibrated cost}
\label{sec:eval-approx}
An approximate gate pays only if its error is below the hardware error of the gates it saves, about $10^{-3}$ per CNOT. For relative-phase Toffolis this never holds (on class-4 supports the best circuit with three CNOTs loses a nonzero amount, Appendix~\ref{app:cost}), so QCPT admits no approximate Toffoli. The error is continuous and small for rotations, which we test on arithmetic in the Fourier basis, the Draper adder of two $n$-bit registers, dropping every controlled phase with an angle below $2\pi/2^m$ (quantum Fourier transform, addition and inverse transform). The operator norm of the dropped gates is at most the sum of their angles, which is the certificate $\epsilon$ of \S\ref{sec:soundness}. We rank the cutoffs by the ideal success probability times $\exp(-C)$, where $C$ is the calibrated cost $\sum -\ln(1-\varepsilon_e)$ of the circuit routed on \texttt{ibm\_aachen} (mean cost of eight transpiler seeds, calibration of 29 October 2026, 03:20). The model prefers the exact adder for $n=4$ ($0.889$, against $0.881$ for $m=3$), the cutoff $m=4$ for $n=5$ ($0.816$ against $0.811$) and for $n=6$ ($0.730$, against $0.719$ for $m=5$ and $0.711$ for the exact adder): the decision depends on the budget and the gain grows with $n$. On hardware ($n=6$, $6$ input pairs, $4000$ shots, $213$--$268$ routed CZ) we compiled the circuits twice with different seeds on a region of $12$ qubits found by the probe, which puts the compilations in different places of the region. Table~\ref{tab:approx} lists the success probabilities. In the first compilation $m=5$ beats the exact adder on $1$ of $6$ inputs and in the second on all $6$ (mean gain $-5.5$ and $+3.5$ points); $m=4$ does so on $4$ of $6$ inputs in both compilations ($+2.6$ and $+1.0$ points). The best cutoff depends on the compilation ($m=4$ in the first, $m=5$ in the second), whereas the model predicts $m=4$, and its absolute success ($0.71$ for the exact adder) is above the measured one ($0.54$). The gains of at most $3.5$ points are of the size of the difference between the compilations, so the hardware neither supports nor excludes the decision to approximate or the choice of the cutoff. For Toffoli networks the budget acts as the minimum two-qubit count, because no approximate Toffoli pays.

\begin{table}[t]
\centering\small
\caption{Success probability of the $6$-bit Draper adder on \texttt{ibm\_aachen} ($4000$ shots, $6$ input pairs, a region of $12$ qubits found by the probe) for two compilations with different seeds.}
\label{tab:approx}
\begin{tabular}{lcc}
\toprule
Cutoff & Compilation 1 & Compilation 2 \\
\midrule
exact & $0.546$ & $0.538$ \\
$m=5$ & $0.490$ & $0.573$ \\
$m=4$ & $0.572$ & $0.547$ \\
\bottomrule
\end{tabular}
\end{table}

\subsection{Certified resynthesis}
\label{sec:eval-resynth}
We applied stage 2 of QCPT to $33$ classical circuits on at most ten qubits: $13$ circuits of the generated suites, including ten of the twelve circuits of the primary suite and $20$ small RevLib circuits (figure~\ref{fig:resynth}). A native multi-controlled X and a phase-sensitive oracle are not classical networks. They stay at their stage-1 cost of $42$ CX, which is counted in the totals of the primary suite below. We used the RevLib circuits while developing the resynthesis and did not hold them out, so their savings are in-sample. Every output is certified. On the primary suite resynthesis lowers the two-qubit count from $281$ (exact) to $194$ in stage 1 and to $160$ ($-43.1\%$, against $-31.0\%$ without resynthesis). On the $20$ RevLib circuits it lowers it from $522$ to $388$ in stage 1 and to $222$ ($-57.5\%$, against $-25.7\%$). A longer search (a budget of $25$\,min per circuit, wider beams, other heuristic weights, exploration slots and noise) lowers six circuits further, all with the same certificate and counted in these totals: \texttt{rd53} $33\to23$, the controlled adder (3b) $31\to27$, the adder with a live carry $14\to13$, and \texttt{4mod5-v1\_23}, \texttt{rd32\_271} and \texttt{rd32\_272} by one CX. Four of the six gains appear within one minute and \texttt{rd53} after $17$\,min (the controlled adder (3b) has no log); the other circuits did not improve, but the search is heuristic and $7$ of the $31$ logged circuits stopped at the budget of $25$\,min before all configurations had run and $5$ stopped after an improvement (the controlled adder (3b) has no log), so it proves no optimality. The mirror tests of \S\ref{sec:eval-hw} run the final circuits of this search, or their placed and routed variants. The C++ core reaches the best known cost on $32$ of the $33$ circuits, including the six above. The median time is $1.5$\,s to the best circuit (over the $28$ circuits where the race improves on its start) and $8.9$\,s for the whole race (maximum $79$\,s for the controlled adder 3b; for the $20$ RevLib circuits $7.8$\,s and $41$\,s). The race stops when no member finds a cheaper circuit in a round of $60$\,s. The race alone leaves \texttt{hwb4\_52} at $26$ CX. Windowed resynthesis (windows of up to four qubits, applied to the circuit of the race, on the C++ core: $51$\,s for all $33$ circuits and at most $28$\,s for one) lowers it to $17$ and four more circuits: \texttt{4\_49\_17} $35\to22$ from the exact circuit, the adder with a live carry $13\to12$, \texttt{alu-v3\_34} $15\to14$ and \texttt{mod5d2\_64} $11\to10$; these gains are in the totals above. Every circuit changed by the race or by the windows passes the same certificate and an independent check against the exact circuit (largest deviation $6\times10^{-15}$). The largest gains occur on arithmetic (controlled adder $30\to17$, modular incrementer $15\to10$, counting each relative-phase Toffoli as $3$ CX); the template search lowers five circuits beyond the beam search: ripple-carry adder $21\to15$, \texttt{alu-bdd} $15\to11$, \texttt{4mod5-bdd} $12\to10$, adder with a live carry $16\to14$ and \texttt{one-two-three-v3} $15\to14$.

\begin{figure*}[t]
\centering
\includegraphics[width=0.92\textwidth]{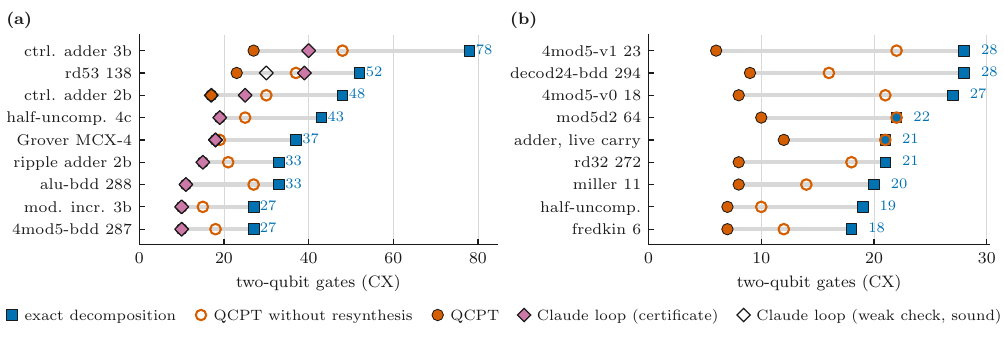}
\caption{\textbf{Two-qubit gates after resynthesis.} Per circuit: exact decomposition (square), QCPT without resynthesis (open circle), QCPT (filled circle), and the best circuit of the search with Claude with the certificate (filled purple diamond) and with the weak check (open diamond; only circuits that pass the certificate afterwards, which the search could not use). \textbf{(a)}~Larger circuits. \textbf{(b)}~Smaller circuits.}
\label{fig:resynth}
\end{figure*}

On the circuits of this section ten random seeds give identical circuits; the verdict of a time-limited certificate and the first finisher of the race of stage 2 can depend on the run.

\paragraph{Comparison with an iterative search by a language model.} As a case study, we compared QCPT with a search in which Claude proposes rewrites of the exact circuit for up to $30$ iterations ($14$ to $30$, on average $25$; one run per circuit) and a verifier accepts them, once with our certificate and once with a check state by state (the weak check), which accepts the counterexample of \S\ref{sec:soundness}. With the certificate the search reached $18$, $25$ and $39$ two-qubit gates on three circuits (compute/uncompute, arithmetic with ancillae, RevLib), against $18$, $17$ and $23$ for QCPT (figure~\ref{fig:resynth}). Over the nine circuits on which both ran QCPT is better on three, equal on six and worse on none; we fixed the nine circuits and the best known cost of QCPT after we had seen these runs, so the comparison is not blind. With the weak check $1$ of the $14$ circuits that the search accepted was not equivalent as a subroutine;  The search reaches its best circuit after $6.5$\,min on average, at $0.77$\,USD and about $5.1\times10^4$ output tokens ($21$\,min for all iterations). Oviedo et al.~\cite{Oviedo2026} estimate a median of $0.31$\,Wh per query for $300$ output tokens and $3.91$\,Wh for $5000$ on H100 nodes, i.e., $0.78$--$1.03$\,mWh per output token and $40$--$52$\,Wh to the best circuit; this is an estimate for a generic deployment, not a measurement of Claude, and it may understate the energy of the prefill of the $9\times10^3$ input tokens of each call. QCPT reaches the same or a better circuit after on average $8.7$\,s and $0.32$--$0.64$\,Wh (CPU time at $150$--$220$\,W of package power shared by $28$ threads, raised by a factor of $1.1$--$1.5$ for supply, board and memory), i.e., about $45$ times faster and $62$--$163$ times less energy. At equal quality, i.e., to the cost at which the search ends, QCPT needs on average $1.6$\,s (median $1.0$\,s) and $0.04$--$0.08$\,Wh, about $250$ times faster and $510$--$1350$ times less energy; on \texttt{rd53} stage 1 alone is already below that cost. These times start from the circuit of stage 1, which takes on average $0.03$\,s ($0.03$\,core-s; at most $0.10$\,s) including its certificate on these nine circuits, so it changes none of the figures above. The time of QCPT for each of the $33$ circuits is stored with the figure script \texttt{make\_times.py} in the release. The two approaches are complementary (figure~\ref{fig:compare}): the search proposes rewrites outside the rules of QCPT, but it is sound only as far as its verifier is, and the certificate of QCPT can serve as that verifier.

\begin{figure}[!tb]
\centering
\includegraphics[width=0.92\columnwidth]{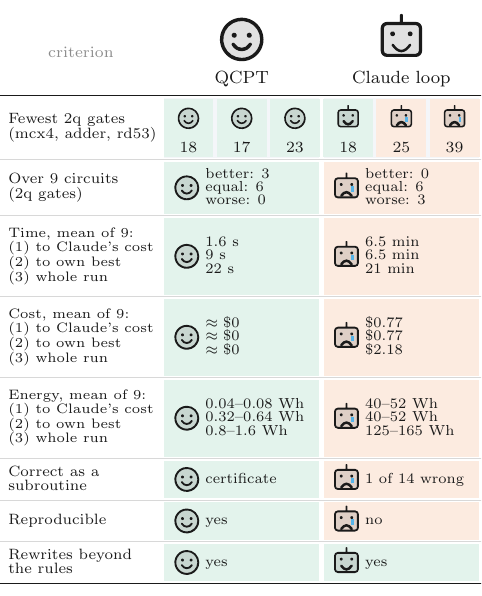}
\caption{QCPT (human icon) against the iterative search with Claude (robot icon). A happy face marks the better result in each criterion (both if tied). Gates: best of up to $30$ iterations with the certificate as the verifier on three circuits (compute/uncompute, arithmetic with ancillae, RevLib; with the weak check, filtered afterwards, $18$, $17$ and $30$) and over nine circuits. Time, cost and energy are means over the nine circuits at three levels: (1) up to the cost at which the search ends, (2) up to the best circuit of each approach and (3) the whole run. The energy of the search is an estimate from published inference measurements~\cite{Oviedo2026,Elsworth2025} ($0.78$--$1.03$\,mWh per output token); that of QCPT is estimated from the CPU time and the nominal power of the processor, raised by $1.1$--$1.5$. The weak check state by state let through wrong circuits. Icons: Phosphor Icons (MIT license).}
\label{fig:compare}
\end{figure}

\subsection{Scalability of verification}
\label{sec:eval-scal}
Verification dominates the cost of QCPT. The per-gate check of condition (R) is an $8\times8$ identity on the projected reachable subspace and does not grow with the circuit, while the whole-circuit certificate requires a dense unitary with $\Theta(4^n)$ memory and time, so QCPT switches to the decision-diagram engine QCEC~\cite{QCEC2021} above $12$ qubits. With the clean ancillae declared as $\ket0$, QCEC decides equivalence up to a global phase on the input subspace, which coincides with subroutine equivalence. It can fail to certify a true equivalence but, up to its numerical tolerance, never certifies a false one, so a failed or timed-out check only makes QCPT use the exact decomposition. On the largest circuits the verdict can differ between runs (parallel checkers under a time limit, random stimuli), which changes the gain, never the correctness. We accept only `equivalent' and `equivalent up to global phase', with the ZX checker and permutation elision off, and every QCEC verdict in the paper agrees with a simulation on three Haar-random inputs.

\begin{figure*}[!tb]
\centering
\includegraphics[width=0.92\textwidth]{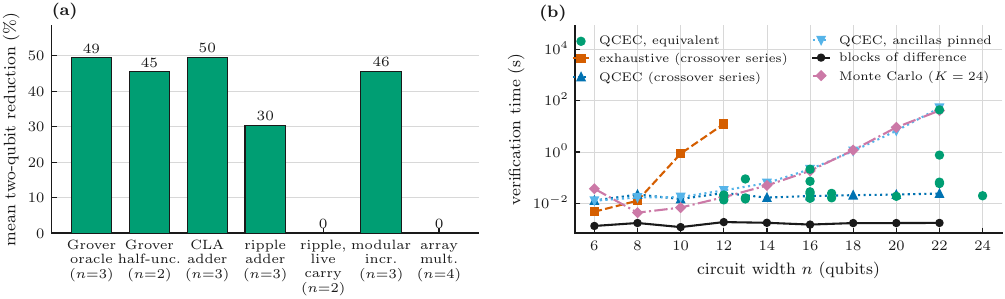}
\caption{Scale evaluation on the scale suite ($12$--$24$ qubits).
(a)~Mean two-qubit-gate reduction by circuit family, with the number of
circuits per family. (b)~Verification
time against circuit width. Dots: QCEC certificate of each of the $19$ outputs it certifies (median $0.02$\,s). Lines: the crossover series, i.e., a ripple adder of
growing width checked with the exhaustive check (squares, dashed; $12$\,s at $n=12$, not attempted beyond), with QCEC (up-triangles, dotted),
with QCEC with the ancillas declared as $\ket0$ (down-triangles), with the certificate by blocks of difference (black dots) and with the
Monte Carlo certificate with $K=24$ vectors (diamonds). Bars of height $0$ in
(a) are families in which no pair, window or (R) substitution applies.}
\label{fig:scale}
\end{figure*}

\paragraph{Verification at scale.} Every output of the scale suite passes the certificate, and none returns the exact decomposition because a certificate failed or timed out. QCEC certifies $19$ circuits with a median time of $0.02$\,s (the $22$-qubit ripple adder needs about $50$\,s on an otherwise idle machine) and the remaining $12$-qubit circuit is certified densely in $0.04$\,s with the C++ columns of the release ($27$\,s with the numpy certificate). For a ripple adder the exhaustive check grows from $0.9$\,s at $n=10$ to $12$\,s at $n=12$, while QCEC needs under $0.025$\,s up to $n=22$ on a ripple adder without ancilla constraint and $6.7$\,s at $n=20$ and about $50$\,s at $n=22$ when the ancillas are declared as $\ket0$ (figure~\ref{fig:scale}b). The release offers two further certificates. The certificate by blocks of difference between the output and the exact circuit is deterministic: it certifies the outputs without a substitution and one with substitutions in at most $0.6$\,s and the ripple adder without ancilla constraint in $2$\,ms at every width up to $n=22$, and it is undecided where the equality holds only on the subspace of the pinned ancillas. The Monte Carlo certificate, opt-in, propagates $K=24$ random vectors and accepts a pair that is not equal up to one phase (Frobenius deviation above $10^{-8}$) with probability below $10^{-14}$. It certifies all $20$ outputs, in at most $0.3$\,s up to $17$ qubits, but in $27$--$37$\,s at $22$ and $262$\,s at $24$ qubits, so on this suite QCEC stays the fastest for wide circuits.

In program mode stage 1 also admits relative-phase Toffolis under (U) at $112$ live Toffolis and removes $46.6\%$ of the two-qubit gates ($32.8\%$ of the infidelity), against $30.1\%$ ($22.2\%$) in subroutine mode, but as subroutines $8$ of the $20$ outputs change the unitary. Count-greedy substitution removes the same $46.6\%$ but changes $11$ outputs, so the gate count alone cannot tell a correct output from a wrong one.

%=======================================================================
\section{Conclusions}
\label{sec:conclusion}
%=======================================================================
QCPT rests on three ideas. First, admissibility is one operator identity with a single phase on the reachable subspace, so the output acts as the exact circuit on every valid input, superposed inputs included, and can be used as a subroutine. Second, one whole-circuit certificate gates both the per-gate selection and the resynthesis of small classical circuits. Stage 1 removes $31.0\%$ of the two-qubit gates on the primary suite, $30.1\%$ on the scale suite, $21.9\%$ on $56$ RevLib circuits and $24.5$--$27.9\%$ on state-resetting circuits~\cite{QuantumResetting2025}, following the role of the Toffolis (Proposition~\ref{prop:cost}), and resynthesis raises this to $43.1\%$ on the primary suite and, on $20$ small RevLib circuits, from $25.7\%$ to $57.5\%$ (in-sample). Third, the accounting bound separates the error of compilation from that of a depolarizing two-qubit model of the device, and on \texttt{ibm\_aachen} QCPT, taking the best of four candidate compilations, returns more often than the Qiskit transpilation of the same circuit in all nine cases and in all $16$ circuits of four families.

QCPT applies to Toffoli networks and classical subcircuits, not to Hamiltonian simulation circuits without Toffolis. On hardware the gain is not uniform (from $+0.008$ for \texttt{grover\_oracle\_5c}, $z=1.2$, to $+0.165$ for \texttt{4mod5-bdd\_287}), the calibration of the device does not show failed couplers, so a region has to be probed before a measurement, and the corrected noise model reproduces the overestimate of the return probability but not the gains of single circuits (correlation $-0.39$). Resynthesis was run up to twelve qubits (about ten for the template search) and costs a median of $8.9$\,s per circuit on one processor (all $33$ circuits) and a budget of $25$\,min for the longer search of \S\ref{sec:eval-resynth}.

Two steps widen the scope (the numbers come from stored runs and are reproduced by the script in the supplement). First, a circuit that arrives as CX, T and H gates hides its Toffolis: replacing each stretch of a three-qubit block whose unitary is a classical permutation by its cheapest circuit over \texttt{x}, \texttt{cx} and \texttt{ccx}, and resynthesizing every classical block of up to five qubits by stage 2, lowered ten classical circuits compiled by Qiskit from $239$ to $156$ CX (basis cx/u, level~$3$), against $152$ for the known network. Second, each CNOT of a routing SWAP is a Toffoli of class $1$ in Proposition~\ref{prop:cost} with the second control fixed to $1$, so the reachable-support analysis applies: on four circuits routed by SABRE QCPT replaced $3$ of $13$ SWAPs and, with windowed resynthesis on coupling edges, lowered $23$, $20$, $22$ and $24$ CX to $22$, $18$, $15$ and $21$, whereas on $19$ scale-suite circuits only $38$ of $674$ SWAPs ($1.0\%$) have a constant qubit.

Further steps are substitutions outside mirrored pairs, tighter over-approximations of $\Reach$, the cost classes of Proposition~\ref{prop:cost} that QCPT does not use yet, certified bounded-approximate decompositions (Appendix~\ref{app:time}), multi-controlled gates~\cite{CleanAncillae2025} and noise-aware routing~\cite{tket2020}.

%=======================================================================
\section*{Code and data availability}
The following are open-source (MIT license) at \mbox{\url{https://github.com/barkol/qcpt}}: the optimizer, the exact and scalable equivalence verifiers, the reachable-subspace admissibility checks, the resynthesis module with the final circuits, all experiment drivers (including the probe \texttt{aachen\_probe.py} and the choice of regions \texttt{select\_from\_probe.py}, the hardware drivers \texttt{aachen\_repro\_hw.py} and \texttt{aachen\_next\_hw2.py}, the calibration check \texttt{aachen\_calcheck.py}, the adder \texttt{qft\_adder\_probed.py} with \texttt{qft\_ranking.py}, the candidates \texttt{qcpt\_topo.py} and \texttt{qcpt\_post.py}, the simulations and noise-model fits in \texttt{experiments/sim\_2026-10-29}, the certificate timings \texttt{scale\_cert\_times.py} and the audit \texttt{audit\_qiskit\_pub.py}), the raw counts and job states of every run on \texttt{ibm\_aachen}, the per-circuit results of the language-model search and of the QCPT race behind the comparison of \S\ref{sec:eval-resynth}, and the figure and analysis scripts with their data. The release (commit \texttt{4efcddf}) implements the subspace conditions of \S\ref{sec:soundness} and accompanies this submission. The search core is the C++17 program in \texttt{cpp/} (\texttt{make}; a compiler and the standard library suffice). The script \texttt{paper/reproduce.sh} rebuilds the tables, figures and numbers of that release from the stored outputs, and the README of \texttt{paper/figures\_resynthesis} maps each figure of the resynthesis and hardware sections to its script and data. The intermediate pickle files of the hardware drivers are regenerated by the drivers and are not part of the release. The decision-diagram certificates use MQT~QCEC~3.6.1~\cite{QCEC2021}; the re-certification of all QCEC results described in \S\ref{sec:eval-scal} is stored in \texttt{paper/recert}, and \texttt{paper/teoria} reproduces Proposition~\ref{prop:cost} and contains the full proofs of its lower bounds.

%=======================================================================
\section*{Author contributions}
K.B.\ conceived the project, developed the admissibility conditions and the
soundness proofs, implemented the optimizer and the evaluation, performed the
hardware experiments, and wrote the manuscript. P.T.\ participated in the
conceptual foundations, worked on the methods for decomposing Toffoli gates,
ran tests, verified the implementation and the benchmark results against the
manuscript, checked the hardware error rates and the references, and
contributed to writing the manuscript. Both authors reviewed the manuscript.
The authors used a large language model (Claude, Anthropic) for language
editing of the manuscript and for checking the bibliography and citations, and
as an assistant in writing and testing the code of QCPT and of the evaluation
and figure scripts. The
scientific content, the derivations and numbers, and the conclusions are the
authors' own, and they take full responsibility for the content.

%=======================================================================
\section*{Acknowledgements}
We thank J\k{e}drzej Stempin for discussions on the necessity of
circuit-depth reduction for quantum-synchronization experiments.
We thank Piotr Biskupski (IBM) for suggestions on the strategy of the hardware experiments, and Bartosz Naskr\k{e}cki for the inspiration to use the Claude autoresearch loop.
The authors acknowledge support from the EuroHPC JU under Horizon Europe
Grant No.\ 101194322 (QEC4QEA), co-funded by the Polish National Centre for
Research and Development (NCBiR) under Decision No.\
DWM/EuroHPC/2023/429/2025.
We acknowledge the use of IBM Quantum services for this work. The views
expressed are those of the authors and do not reflect the official policy or
position of IBM or the IBM Quantum team.

%=======================================================================
\bibliographystyle{refstyle}
\bibliography{refs}
%=======================================================================

\appendix

\section{Error model of the hardware test}
\label{app:hw}
\paragraph{Model.} Three optional terms extend the model of \S\ref{sec:bg-fidelity}: idle errors (Pauli-twirled amplitude and phase damping of every active qubit under an as-late-as-possible schedule, with $p_x=p_y=(1-e^{-\tau_q/T_1})/4$ and $p_z=(1-e^{-\tau_q/T_2})/2-p_x$), static $ZZ$ cross-talk (phase $2\pi\zeta T$ on $\ket{11}$ of two coupled active qubits, twirled to the factor $1-\tfrac34\sin^2(\phi/2)$) and an excess two-qubit error $(1-\lambda)^{n_{2q}}$ for leakage and the gap between the isolated and the in-circuit error. They change the value of the budget but no decision of QCPT, because every admitted substitution is also shorter than the exact decomposition before transpilation. The corrected simulation model of \S\ref{sec:eval-hw} is the density-matrix counterpart of these terms with sizes fitted to the single-batch data: its factor for the CZ error, $K_2=1.94$, agrees with the median ratio of $2.0$ to the calibration (range $1.4$--$3.3$ over eight couplers) that randomized benchmarking gives, whereas $R_2$ makes $T_2$ $14$ times shorter against the measured $4.8$, so the fit also absorbs other effects.

\section{Characterization of the device}
\label{app:bench}
To check the manufacturer's calibration we measured \texttt{ibm\_aachen} ourselves on 29 October 2026. On $14$ isolated qubits $T_1$ and the echo $T_2$ are $1.09$ and $1.05$ of the calibration and the readout error $1.14$; the error per gate from randomized benchmarking is $1.6$ times the calibration for single-qubit gates and $2.0$ times for the CZ gate ($3.0$ times before converting the error per Clifford to the error per gate; eight couplers of the regions used, benchmarked simultaneously). The echo $T_2$ with excited neighbors, the regime of the mirror circuits, is only $0.21$ of the calibration. The randomized benchmarking and the coherence measurements used two seeds and $1000$ shots each, so their ranges are wider than those of a full benchmark.

\paragraph{Routes of the candidate with its own routing.} It won \texttt{ripple\_carry\_adder\_2b}, \texttt{grover\_oracle\_5c}, \texttt{mod\_incr\_5b} and \texttt{ripple\_adder\_4b}, with $23$, $62$, $32$ and $74$ two-qubit gates for QCPT's router, $31$, $59$, $26$ and $93$ for SABRE on stage~2 and $29$, $45$, $37$ and $57$ for SABRE on stage~1 (from the unreleased intermediate files of the driver, recorded in \texttt{experiments/fixtures/trasy\_qct\_topo.json}).

\begin{table*}[t]
\centering\footnotesize
\caption{Hardware points on probed regions of \texttt{ibm\_aachen} ($8000$ shots each, standard error of $P$ at most $0.005$): the nine circuits of figure~\ref{fig:hw} and the scaling families of figure~\ref{fig:dyno}, measured in one job, and five circuits marked $\dagger$ that we fixed in advance and measured in an earlier job (\texttt{ripple\_adder\_4b} and four RevLib circuits not measured before). 2q gives the two-qubit gates of the mirror circuit after transpilation for Qiskit and for QCPT, $P$ the return probability and $z$ the gain in standard errors of the shot noise. The standard errors count only the shot noise: the control \texttt{rl\_4mod5-v0\_18}, measured in both jobs an hour apart on the same qubits, gave $0.803\to0.861$ and $0.762\to0.830$.}
\label{tab:probed}
\setlength{\tabcolsep}{5pt}
\begin{tabular}{lccccc}
\toprule
Circuit & 2q & $P$ Qiskit & $P$ QCPT & gain & $z$ \\
\midrule
\texttt{rl\_4mod5-bdd\_287} & $104/77$ & $0.456$ & $0.621$ & $+0.165$ & $21.2$ \\
\texttt{half\_uncomputed\_4c} & $186/157$ & $0.434$ & $0.572$ & $+0.139$ & $17.7$ \\
\texttt{mod\_incr\_4b} & $138/91$ & $0.555$ & $0.676$ & $+0.121$ & $15.8$ \\
\texttt{rl\_4mod5-v1\_23} & $88/64$ & $0.714$ & $0.814$ & $+0.100$ & $15.0$ \\
\texttt{rl\_decod24-bdd\_294}$^\dagger$ & $88/67$ & $0.717$ & $0.815$ & $+0.097$ & $14.6$ \\
\texttt{rl\_alu-bdd\_288}$^\dagger$ & $132/107$ & $0.606$ & $0.701$ & $+0.095$ & $12.7$ \\
\texttt{grover\_oracle\_mcx4} & $139/114$ & $0.546$ & $0.639$ & $+0.093$ & $12.1$ \\
\texttt{ripple\_adder\_3b} & $184/153$ & $0.293$ & $0.381$ & $+0.088$ & $11.8$ \\
\texttt{ripple\_carry\_adder\_2b} & $120/86$ & $0.634$ & $0.712$ & $+0.078$ & $10.6$ \\
\texttt{rl\_4mod5-v0\_18} & $83/55$ & $0.762$ & $0.830$ & $+0.069$ & $10.8$ \\
\texttt{rl\_hwb4\_52} & $111/84$ & $0.766$ & $0.832$ & $+0.066$ & $10.4$ \\
\texttt{rl\_one-two-three-v2\_100} & $85/60$ & $0.732$ & $0.795$ & $+0.063$ & $9.4$ \\
\texttt{rl\_rd32\_271} & $101/67$ & $0.746$ & $0.806$ & $+0.061$ & $9.2$ \\
\texttt{half\_uncomputed\_min} & $53/34$ & $0.854$ & $0.915$ & $+0.061$ & $12.1$ \\
\texttt{half\_uncomputed\_3c} & $111/84$ & $0.679$ & $0.735$ & $+0.056$ & $7.8$ \\
\texttt{ripple\_adder\_1b} & $46/32$ & $0.860$ & $0.910$ & $+0.050$ & $9.9$ \\
\texttt{grover\_oracle\_2c} & $40/28$ & $0.870$ & $0.918$ & $+0.048$ & $9.8$ \\
\texttt{rl\_4\_49\_17}$^\dagger$ & $128/93$ & $0.771$ & $0.817$ & $+0.046$ & $7.2$ \\
\texttt{half\_uncomputed\_2c} & $75/51$ & $0.838$ & $0.874$ & $+0.037$ & $6.6$ \\
\texttt{ripple\_adder\_4b}$^\dagger$ & $287/200$ & $0.141$ & $0.176$ & $+0.036$ & $6.2$ \\
\texttt{rl\_one-two-three-v3\_101}$^\dagger$ & $92/68$ & $0.745$ & $0.779$ & $+0.034$ & $5.0$ \\
\texttt{mod\_incr\_5b} & $216/128$ & $0.299$ & $0.331$ & $+0.032$ & $4.3$ \\
\texttt{mod\_incr\_2b} & $38/30$ & $0.887$ & $0.914$ & $+0.027$ & $5.6$ \\
\texttt{mod\_incr\_3b} & $88/51$ & $0.772$ & $0.797$ & $+0.025$ & $3.9$ \\
\texttt{grover\_oracle\_mcx3} & $105/79$ & $0.771$ & $0.793$ & $+0.022$ & $3.3$ \\
\texttt{adder\_live\_carry} & $72/68$ & $0.828$ & $0.845$ & $+0.018$ & $3.0$ \\
\texttt{grover\_oracle\_5c} & $196/142$ & $0.244$ & $0.252$ & $+0.008$ & $1.2$ \\
\bottomrule
\end{tabular}
\end{table*}

\section{Running time under a noise budget}
\label{app:time}
We derive when a shorter, approximate circuit reduces the running time of an experiment that estimates $\langle O\rangle$ for a Pauli observable to a mean-square error $\delta^2$. Let the exact circuit have $n$ two-qubit gates in the light cone of $O$, each damping the signal by $e^{-\varepsilon}$ (for two-qubit depolarizing noise, $\varepsilon=-\ln(1-4\varepsilon_{\mathrm{avg}}/3)$, since $8$ of the $15$ Pauli errors anticommute with $O$), and let a shot take $t_0+d\tau$ (fixed time $t_0$ of reset and readout, depth $d$, layer time $\tau$). The estimator is $\hat S\,e^{n\varepsilon}$ with $\hat S$ the mean of $N$ results $\pm1$, variance $(1-S^2)/N$ and $S=\langle O\rangle e^{-n\varepsilon}$. A circuit with $n-\Delta$ gates and infidelity $1-F$ to the exact one adds a bias of at most $\eta=2\sqrt{1-F}$ for a pure input state, so the error $\delta^2$ requires $N'=e^{2(n-\Delta)\varepsilon}(1-S'^2)/(\delta^2-\eta^2)$ shots ($\eta<\delta$), and $T'/T=\frac{e^{-2\Delta\varepsilon}}{1-\eta^2/\delta^2}\cdot\frac{1-S'^2}{1-S^2}\cdot\frac{t_0+d'\tau}{t_0+d\tau}\approx\frac{e^{-2\Delta\varepsilon}}{1-\eta^2/\delta^2}$ for $t_0\gg d\tau$. It is faster if $\Delta\varepsilon>-\tfrac12\ln(1-\eta^2/\delta^2)\approx\eta^2/(2\delta^2)$, i.e., if the noise removed with the gates exceeds half of the fraction of the mean-square budget spent on the approximation. For $n\varepsilon=1$, removing $10\%$ of the gates without bias gives $T'/T\approx0.82$ in this approximation and $0.79$ with the factor $(1-S'^2)/(1-S^2)$ for $\langle O\rangle=1$; with $\eta^2/\delta^2=0.1$ it gives $0.91$ and $0.88$. In our experiments the circuit lasts below $10\%$ of a shot; the model assumes depolarizing noise and no mitigation that changes the variance.

\section{The primary suite}
\label{app:suite}
Table~\ref{tab:suite} lists the $12$ circuits of the primary suite with the two-qubit counts of every method. Count-greedy substitution changes the unitary of the six circuits marked $\dagger$ (as subroutines), QCPT of none. Native multi-controlled and phase-sensitive circuits contain no Toffoli that QCPT can change.

\begin{table}[t]
\centering\footnotesize
\caption{Primary suite: qubits $n$, Toffolis (CCX), two-qubit gates of the exact decomposition, stage 1, QCPT (both stages) and count-greedy substitution ($^\dagger$: not equivalent as a subroutine). Totals: $281$, $194$, $160$, $167$.}
\label{tab:suite}
\setlength{\tabcolsep}{2.2pt}
\begin{tabular}{@{}lcrrrrr@{}}
\toprule
Circuit & $n$ & CCX & exact & stage 1 & QCPT & greedy \\
\midrule
Ripple adder (2b)         & $6$ & $4$ & $33$ & $21$ & $15$ & $21$ \\
Controlled adder (2b)     & $7$ & $8$ & $48$ & $30$ & $17$ & $24^\dagger$ \\
Grover, MCX-3            & $6$ & $4$ & $25$ & $13$ & $12$ & $13$ \\
Grover, MCX-4            & $8$ & $6$ & $37$ & $19$ & $18$ & $19$ \\
Grover, native MCX        & $5$ & $0$ & $36$ & $36$ & $36$ & $36$ \\
Comp./uncomp. pair       & $4$ & $2$ & $13$ & $7$ & $7$ & $7$ \\
Nested comp./uncomp.     & $6$ & $4$ & $25$ & $13$ & $12$ & $13$ \\
Live AND chain            & $5$ & $2$ & $12$ & $12$ & $12$ & $6^\dagger$ \\
Phase-sensitive oracle   & $3$ & $1$ & $6$ & $6$ & $6$ & $3^\dagger$ \\
Single live Toffoli       & $3$ & $1$ & $6$ & $6$ & $6$ & $3^\dagger$ \\
Half-uncomputed           & $5$ & $3$ & $19$ & $10$ & $7$ & $10^\dagger$ \\
Adder, live carry         & $5$ & $3$ & $21$ & $21$ & $12$ & $12^\dagger$ \\
\midrule
Total & & & $281$ & $194$ & $160$ & $167$ \\
\bottomrule
\end{tabular}
\end{table}

\section{Rewrite rules for Toffoli networks}
\label{app:rules}
Before stage 1, QCPT applies three rewrite rules to classical networks (\texttt{x}, \texttt{cx}, \texttt{ccx}), each an identity on bits. \emph{Common control:} $\mathrm{ccx}(a,b,t)\,\mathrm{ccx}(a,c,t)$ adds $a(b\oplus c)$ to $t$, as does $\mathrm{cx}(c,b)\,\mathrm{ccx}(a,b,t)\,\mathrm{cx}(c,b)$, provided no gate in between touches $a,b,c,t$ (cost $12\to8$). \emph{Compute, copy, uncompute:} $\mathrm{ccx}(a,b,t)\,\mathrm{cx}(t,c)\,\mathrm{ccx}(a,b,t)$ adds $t\oplus ab$ to $c$ and leaves $t$, as does $\mathrm{cx}(t,c)\,\mathrm{ccx}(a,b,c)$ for $c\notin\{a,b\}$ ($13\to7$). \emph{Same AND:} if $t_1$ holds $0$ before $\mathrm{ccx}(a,b,t_1)$ and nothing in between touches $a$, $b$ or $t_1$, a later $\mathrm{ccx}(a,b,t_2)$ equals $\mathrm{cx}(t_1,t_2)$ for any $t_2$ ($6\to1$). A Margolus gate whose middle control is the parity of two qubits equals $\mathrm{CX}(q_2,q_1)\,G(q_1,b,t)\,\mathrm{CX}(q_2,q_1)$ and costs one CNOT less with the CNOTs fused into the gate ($4$-CNOT parity Margolus gates of the library). We verified the identities exhaustively (all bit assignments, unitary equality up to a phase) and the implementation on all $45\,528$ classical circuits with at most three gates on four qubits (two sets of pinned qubits), with no wrong output. Stage 1 runs on every non-empty subset of the rules and keeps the cheapest result, so the rules cannot raise the count (without this choice they raised it for the controlled adders, $30\to31$ and $48\to50$); in a separate run they lower the stage-1 count on $8$ of $33$ circuits , so the gain is empirical.

\section{Calibration-aware orientation}
\label{app:orient}
Both gates apply $\mathrm{CX}(b,t)\,\mathrm{CX}(a,t)\,\mathrm{CX}(b,t)$, so one control--target pair carries two CNOTs, and the choice of $b$ does not change the count. For per-edge errors from a calibration, QCPT places the doubled CNOT on the cheaper pair. Exchanging the controls moves the phase to another basis state, so a Margolus gate admitted under (R) keeps the new orientation only if it passes (R) again, the gates of a pair or window are oriented together, and an oriented circuit that fails the whole-circuit certificate reverts to the default.

We applied this step to the primary suite and to RevLib with four device snapshots (three IBM devices, one of them a day earlier) and five routing seeds. The optimal orientation differs for $38$ of $165$ gates ($23\%$) between the three devices and $19$ ($12\%$) between two days. In the $113$ cases with a reoriented gate the calibrated infidelity after routing falls by $0.26\%$ in total ($73$ decreases, $38$ increases; Wilcoxon $p=0.054$; $2$ reoriented circuits failed the certificate). The orientation was not tested on hardware in the series of 29 October.

\section{Optimal cost of a Toffoli on a support}
\label{app:cost}
\paragraph{Constructions.} In class $0$ the Toffoli acts on $L(g)$ as $I$ or $X_t$, in class $1$ one control is redundant, and in class $2$ the circuit $X_t\,\mathrm{CX}(b,t)\,\mathrm{CX}(a,t)$ is exact on $L(g)$. In class $3$ the Margolus gate, the mirrored Margolus gate and the two gates with exchanged controls have the phase $-1$ on $100$, $101$, $010$ and $011$ and satisfy (R) when that string is missing; in class $4$ four CNOTs with suitable single-qubit gates suffice for each of the eight supports (found numerically to $10^{-9}$, in the repository). 

\paragraph{Lower bounds.} Since $c^*$ is monotone in $L$, the smallest supports of each class suffice. At most one CNOT, or two on the same pair, leave a qubit untouched (values $1$ and $2$); two CNOTs on different pairs form a chain whose cycle through the control values has an odd number of orthogonal pairs because $a\wedge b$ is not of the form $c(b)\oplus d(a)$ ($3$); for $4$, a spectral argument across the cut between target and controls and Makhlin invariants~\cite{Makhlin2002} exclude fewer CNOTs. Numerical optimization with the given number of CNOTs and arbitrary single-qubit gates on the minimal supports ($12$, $8$, $12$ of $16$ and $4$ for the classes $1$ to $4$) gives a nonzero loss in every case; for the classes $3$ and $4$ these are numerical checks, not proofs.

\paragraph{Static prediction.} With cost $3$ for members of pairs and windows and the classes $0$, $1$, $3$ for the others, the predicted two-qubit count equals that of the released code for all $82$ circuits ($528$ Toffolis; the swapped orientation is not tried, \S\ref{sec:method-cost}).

\end{document}